\documentclass[12pt]{article}
\usepackage{fullpage}
\usepackage{setspace}
\usepackage{multirow}
\usepackage{makecell}
 \usepackage{pdflscape}

\usepackage{booktabs}
\usepackage{caption}

\usepackage{amsfonts}
\usepackage[colorlinks,linkcolor=blue, anchorcolor=red,citecolor=blue,CJKbookmarks=true]{hyperref}
\usepackage{epsfig, graphicx, floatrow,rotating, amssymb, cite, bbm, amsthm, mathrsfs, dsfont, makeidx,color}
\usepackage{bm}
\usepackage[tbtags]{amsmath}

\allowdisplaybreaks[4]
\usepackage[sectionbib]{natbib}

\usepackage{algorithm}  
\usepackage{algorithmicx} 
\usepackage{algpseudocode}

\newtheorem{condition}{Condition}
\newtheorem{theorem}{Theorem}
\newtheorem{lemma}{Lemma}
\newtheorem{proposition}{Proposition}
\floatname{algorithm}{Algorithm}

\usepackage{moreverb}

\title{Distributed Conditional Feature Screening via Pearson Partial Correlation with FDR Control}
\author{Naiwen Pang, Xiaochao Xia\footnote{Corresponding author, Email:xxc@cqu.edu.cn} \\
College of Mathematics and Statistics, Chongqing University, 401331}
\date{ }

\begin{document}

\maketitle

\begin{abstract}
This paper studies the distributed conditional feature screening for massive data with ultrahigh-dimensional features. Specifically, three distributed partial correlation feature screening methods (SAPS, ACPS and JDPS methods) are firstly proposed based on Pearson partial correlation. The corresponding consistency of distributed estimation and the sure screening property of feature screening methods are established. Secondly, because using a hard threshold in feature screening will lead to a high false discovery rate (FDR), this paper develops a two-step distributed feature screening method based on knockoff technique to control the FDR. It is shown that the proposed method can control the FDR in the finite sample, and also enjoys the sure screening property under some conditions. Different from the existing screening methods, this paper not only considers the influence of a conditional variable on both the response variable and feature variables in variable screening, but also studies the FDR control issue. Finally, the effectiveness of the proposed methods is confirmed by numerical simulations and a real data analysis.

\end{abstract}

\section{Introduction}
With the rapid development of information technology, it has become easier for researchers to collect, store and manipulate the large-scale/massive data sets with ultrahigh-dimensional features in various fields, such as finance, biology, engineering, social science, medicine and so forth. How to deal with these large-scale data sets quickly and effectively has gradually become one of the important frontier problems in the statistical community.

For ultrahigh-dimensional data, which means the number of covariates(features) $p$ is much larger than sample size, it is usually assumed that only a few features are relevant to the response variable. To handle such data, researchers usually adopt the strategy of filtering out irrelevant features before modeling and forecasting. However, due to the challenges of computational expediency, statistical accuracy, and algorithmic stability (\citep{r12}), the regularization-based variable selection methods are generally difficult to apply to ultrahigh-dimensional problems directly. In order to solve the problems caused by ultrahigh-dimensional data, \citet{r5} proposed a ultrahigh-dimensional feature screening method under Gaussian linear model, namely, the sure independence screening (SIS). Their method uses Pearson correlation as a utility function to rank the importance of features. They further proved that the SIS method enjoys the sure screening property that the set of selected features contains all important features with asymptotic probability one. This property has become a basis theory for evaluating the performance of feature screening methods. To relax the restriction of the SIS method on linear model, \citet{r6} proposed a model-free sure independence ranking screening (SIRS) method. This method constructs a utility function using the conditional distribution of the response variable. \citet{r13} proposed a feature screening method based on distance correlation (DC-SIS). \citet{r7} developed a robust feature screening method named the projection correlation-based screening (PC-SIS) via projection correlation introduced in \citet{r14}. \citet{r15} further studied the feature screening method based on a stable correlation. 
 
In most of the existing studies on ultrahigh-dimensional feature screening, the sample size $N$ is not very large, which is usually of hundreds in DNA sequencing. However, in some real problems on massive data, the data size is very likely to be extremely large or huge. This will bring many challenges in storage and computation. A common solution is to adopt the "divide-and-conquer" strategy. That is, the original problem with large-scale data is firstly divided into multiple small- or moderate-scale data sub-problems, each of which can be solved separately on a single computer, and then the results of the sub-problems are combined to form the final result for the original problem. Naturally, parallel computing can be employed to rapidly solve the sub-problems, thus the computing efficiency is improved. On the other hand, in multi-institution medical research, data may not be shared between institutions partly due to privacy policies. This also encourages researchers to study the computation  method of distributed data or the fusion technology of multi-source data. For example, \citet{r2} first applied the distributed method based on the "divide and conquer" strategy to the ultrahigh-dimensional feature screening of massive data. They proposed the aggregated feature screening method that can be applied to a wide range of correlation measures, and proved that their distributed feature screening method still enjoys sure screening property. \citet{r1} utilized the subsampling method to screen the features of massive ultrahigh-dimensional data, and also established the sure screening property for their method. \citet{r25} proposed a distributed feature screening method for the categorical response, which measures the importance of features to the categorical response using a conditional rank utility (CRU) function.

The main purpose of feature screening is to select a subset $\hat{\mathcal{M}}=\left\{j:\hat{\omega}_j \geq \gamma,j=1,\ldots, p \right\}$ to contain all important features, where $\hat{\omega}_j$ is an estimator of some commonly used correlation coefficient between the feature $X_j$ and the response $Y$. The key is to select an appropriate threshold, $\gamma$. Theoretically, as long as $\gamma$ satisfies a certain condition, feature screening method enjoys the sure screening property. But, it is not easy to determine an appropriate threshold in practice. In order to ensure sure screening property, a practical choice is to firstly rank the sample correlations and then select the top $d$ features as the set of important features, where $d$ is slightly smaller than the sample size. \citet{r5} suggest to use $d=\lfloor n/\log(n) \rfloor$ in reality, where $n$ denotes the sample size and $\lfloor \cdot \rfloor$ denotes the largest integer not exceeding itself. \citet{r6} determined the threshold by introducing auxiliary variables. However, all of these methods can lead to a high false discovery rate (FDR), indicating that too many noisy features are retained in the selected model. Thus, a desired feature screening method can enjoy the sure screening property and control FDR under a low level simultaneously. Regarding the FDR control, \citet{r8} recently put forward the concept of knockoff variables when considering the variable selection problem of Gaussian linear model. They generated knockoff variables for a fixed design matrix, and proposed a variable selection method that can control FDR. \citet{r9} extended the knockoff method to random design matrix. \citet{r16} further extended the Knockoff method to high-dimensional linear models. More recently, \citet{r7} has extended the knockoff variable method to ultrahigh-dimensional feature screening. Whereas, these methods can not be applicable to massive data. 

On the other hand, as \citet{r5} pointed out that when the dimension $p$ increases, there will be spurious correlations between covariables, so there may be some potential problems with feature screening methods based on marginal correlation. For instance, it may result in some important features that are marginally uncorrelated but jointly correlated with the response not being selected, while some unimportant features are selected because of their strong correlation with important features. And \citet{r17} also pointed out that in many real-world applications, it is necessary for researchers to adjust the correlation for the influence of some variables. For example, in the study of the association between two biomarkers, researchers often want to eliminate the influence of confounding factors such as age, weight, gender and so on. Based on the above reasons, some researchers consider to use partial correlation or conditional correlation as utility function in feature screening. \citet{r18} proposed the PC-simple algorithm, which simplifies the PC algorithm of \citet{r22}. They used partial correlation to measure the relationship between response and variables, and then carried out variable selection. \citet{r19} studied the limit distribution of sample partial correlation under the ellipsoidal distribution hypothesis, and proposed a new feature screening and variable selection method under the ultrahigh-dimensional linear model. \citet{r20} studied the quantile partial correlation-based feature screening in ultrahigh-dimensional quantile linear regression model.  \citet{r24} proposed a robust feature screening method based on conditional quantile correlation. \citet{r23} proposed a partial correlation based on copula function, and proposed a robust feature screening method (CPC-SIS). \citet{r11} proposed a conditional feature screening method based on conditional independence measure, and proved that the method enjoys the sure screening property and rank consistency property under certain conditions. 

In this paper, we propose a distributed conditional feature screening procedure with FDR control, which is an extension of the work of \citet{r2}.  Specifically, we introduce Pearson partial correlation into the distributed feature screening of massive data with ultrahigh-dimension, and prove that the proposed distributed feature screening method enjoys the sure screening property and ranking consistency property under certain conditions. Meanwhile, compared with \citet{r7} for feature screening based on marginal correlation, the method developed in this paper allows the influence of a conditional variable on the response variable. Moreover, three distributed feature screening methods are presented, and the theoretical properties of the methods are established. In addition, a two-step distributed feature screening method based on knockoff features is proposed, which achieves automatic threshold selection by controlling FDR. The theoretical properties of the method are also provided.

The rest of this paper is organized as follows. In Section 2, the methodology is  introduced. In Section 3, some related theoretical properties are presented. In Section 4, we provide a distributed feature screening method based on FDR control and their theoretical properties. Numerical simulations are given in Section 5, and a real-world data set is illustrated in Section 6. Concluding remarks are given in Section 7. Proofs of theorems are given in the Appendix.

\section{Distributed Feature Screening Based on Partial Correlation}

\subsection{Model and Notations}
Let $\mathcal{D}={\left\{(Y_i,\mathbf{X}_i,Z_i)\right\}}^{N}_{i=1}$ be $N$ independent and identically distributed as the population $(Y,\mathbf{X},Z)$, where $Y$ is a response variable, $\mathbf{X}={(X_1,\ldots,X_p)}^{T}$ is a $p$-dimensional feature vector, and $Z$ is the conditional variable. To measure the correlation between $\mathbf{X}$ and $Y$, we want to adjust for the variable $Z$. In high-dimensional data analysis, sparsity assumption is often made, which means only a small set of feature variables are important. We use $\mathcal{M}$ to represent the set of indices of all important features, and its complement is denoted as $\mathcal{M}^{c}=\left\{1,\ldots,p\right\} \setminus \mathcal{M}$. 

The goal of feature screening is to estimate $\mathcal{M}$. Let $\omega_j \geq 0 $ be a measure of correlation strength between $Y$ and $X_j$, then we define $\mathcal{M}=\left\{j:\omega_j > 0,j=1,\ldots,p \right\}$. Let $\hat{\omega}_j$ be an estimator of $\omega_j$ based on the data $\mathcal{D}$. We define
$$ \hat{\mathcal{M}}=\left\{  j:\hat{\omega}_j \geq \gamma,j=1,\ldots,p \right\}, $$
with a pre-specified threshold $\gamma >0$. The Pearson correlation is the most commonly used correlation to measure the linear relationship between two random variables. When the conditional feature $Z$ has influence on both $Y$ and $X_j$, it may not be accurate to describe the association between $X_j$ and $Y$ using only Pearson correlation. Therefore, it is more suitable to use a partial correlation, $\rho_{Y,X_j|Z}$, to measure the association between the response $Y$ and the feature $X_j$ while controlling for $Z$. Set $\hat{\omega}_j=|\hat{\rho}_{Y,X_j|Z}|$, where $\hat{\rho}_{Y,X_j|Z}$ is an estimator of $\rho_{Y,X_j|Z}$. Specifically, let $\tilde{\epsilon}$ and $\tilde{\zeta}_j$ represent the partial residuals obtained by regressing $Y$ and $X_j$ on $Z$, respectively, then
\begin{equation}
\rho_{Y,X_j|Z}=\rho_{\tilde{\epsilon},\tilde{\zeta}_j}=\frac{\mathrm{cov}(\tilde{\epsilon},\tilde{\zeta}_j)}{\sqrt{\mathrm{var}(\tilde{\epsilon})}\sqrt{\mathrm{var}(\tilde{\zeta}_j)}}, \label{0} 
\end{equation}
where $\rho_{A,B} $ denotes the Pearson correlation between two variables $A$ and $B$. Since the partial residuals $\tilde{ \epsilon}$ and $\tilde{ \zeta} _j$ can eliminate the influence of $Z$ on $Y$ and $X_j$, the above Pearson correlation between $\tilde{ \epsilon} $ and $\tilde{ \zeta} _j$ are called Pearson partial correlation, referred to as partial correlation. Note that equation \eqref{0} can be equivalently expressed as
\begin{equation}
\rho_{Y,X_j|Z}=(\rho_{Y,X_j}-\rho_{X_j,Z}\rho_{Y,Z}) / \sqrt{(1-\rho_{X_j,Z}^2)(1-\rho_{Y,Z}^2)}. \label{1}
\end{equation}
Thus, we can obtain $\hat{\rho}_{Y,X_j|Z}$ as
$\hat{\rho}_{Y,X_j|Z}=(\hat{\rho}_{Y,X_j}-\hat{\rho}_{X_j,Z}\hat{\rho}_{Y,Z})/\sqrt{(1-\hat{\rho}_{X_j,Z}^2)(1-\hat{\rho}_{Y,Z}^2)} $, 
where $\hat{\rho}$ represents the sample correlation.

\subsection{Distributed Feature Screening Methods}
\subsubsection{Aggregated Feature Screening}
When $N$ and $p$ are both large, it is difficult to compute the $\left\{ \hat{ \omega} _j\right\}^p_{j=1} $ based on the full data set $\mathcal{D}$. Inspired by the idea of "divide and conquer", we consider a distributed computation method. We divide $\mathcal{D}$ into $K$ disjoint subsets $\left\{ \mathcal{D}_k\right\}^K_{k=1} $, each of which contains $n=N/K$ elements. 
To simplify, we use $\hat{\rho}_{k,j}$ to represent $\hat{\rho}_{Y,X_j|Z}$ estimated based on $\mathcal{D}_{k}$. A simple idea is to compute an averaged correlation estimator
\begin{equation}
 \bar{\omega}_j=|\frac{1}{K}\sum^{K}_{k=1} \hat{\rho}_{k,j}|. \label{4}  
\end{equation}

By retaining the features with the large value of $\bar{\omega}_j$, we could get an estimator of $\mathcal{M}$, denoted as $\bar{\mathcal{M}}$. We name this method as the Simple Average distributed Partial correlation Screening (SAPS). The method based on simple average is easy to operate, but the estimator $\bar{\omega}_j$ is prone to bias, especially when $K$ is large so that the size of each subset of data is small. 

Motivated by \citet{r2}, we propose the following improved distributed estimator of partial correlation. Specifically, from \eqref{1}, the partial correlation $\rho_{Y,X_j|Z}$ can be expressed as a function of some parameters, i.e.,  
\begin{equation}
\rho_{Y,X_j|Z}=g(\theta_{j,1},\ldots,\theta_{j,9}), \label{5}
\end{equation}
where $\theta_{j,1}=E(X_jY), \theta_{j,2}=E(X_jZ),\theta_{j,3}=E(YZ),
\theta_{j,4}=E(X_j), \theta_{j,5}=E(Y), \theta_{j,6}=E(Z),
\theta_{j,7}=E(X^2_j), \theta_{j,8}=E(Y^2)$ and $\theta_{j,9}=E(Z^2)$.

On each subsets $\mathcal{D}_k$, we can estimate $\theta_{j,s} (s=1,\ldots,9) $ by a local U-statistic
\begin{equation}
\hat{\theta}_{j,s}^k=\frac{1}{n}\sum_{i_1 \in \mathcal{S}_k}\hat{\theta}_{j,s}(W_{i_1,j}), \label{6}
\end{equation}
where $\hat{\theta}_{j,s}(W_{i_1,j})$ is the kernel of U-statistic, $W_{j}=(Y,X_j,Z)$, $\mathcal{S}_k$ is the index set of $\mathcal{D}_k$. Then we obtain an estimator of $\omega_j$ as
\begin{equation}
\tilde{\omega}_j=|g(\bar{\theta}_{j,1},\ldots,\bar{\theta}_{j,9})|, \label{7}
\end{equation}
where $\bar{\theta}_{j,s}=\frac{1}{K}\sum^{K}_{k=1}\hat{\theta}_{j,s}^k$. As a result, we can obtain an estimator of $\mathcal{M}$ based on $\tilde{\omega}_j$ as
$$ \tilde{\mathcal{M}}=\left\{  j:\tilde{\omega}_j \geq \gamma,j=1,\ldots,p \right\}.$$
We name this method as the Aggregated Component distributed Partial correlation Screening(ACPS).

\subsubsection{Debiased Feature Screening}
The shortcoming of SAPS is that it could produce a bias. A natural idea is to subtract a value from the simple average estimator to eliminate the bias. Following the method of \citet{r1}, we use the jackknife method to construct a new distributed estimator of the partial correlation.

We use $\mathcal{D}_{k,-i}$ to denote the data set $\mathcal{D}_k$ with the $i$th sample observation being removed. Define
\begin{equation}
\hat{\Delta}_{k,j}=n^{-1}(n-1)\sum^n_{i=1}\hat{\rho}_{k,j,-i}-(n-1)\hat{\rho}_{k,j}, \label{8}
\end{equation}
and let $\hat{\rho}_{k,j,-i}$ be the $\hat{\rho}_{Y,X_j|Z}$ based on $\mathcal{D}_{k,-i}$. Then we obtain another estimator of $\omega_j$ as
\begin{equation}
\bar{\omega}^D_{j}=|\frac{1}{K}\sum^{K}_{k=1}(\hat{\rho}_{k,j}-\hat{\Delta}_{k,j})|. \label{9}
\end{equation}
Accordingly, we obtain an estimator of $\mathcal{M}$ as
$$ \bar{\mathcal{M}}^D=\left\{  j:\bar{\omega}^D_j \geq \gamma,j=1,\ldots,p \right\}, $$
we name this method as the Jackknife Debiased distributed Partial correlation Screening(JDPS).

\section{Theoretical Property}
We now provide some theoretical properties of SAPS, ACPS and JDPS. To this end, we assume the following conditions.
\begin{condition}\label{a1}
There exists two constants $\kappa_0$ and $D_0$ such that, for any $0 \leq \kappa \leq \kappa_0$, for all $s=1,\ldots, 9$ and $j=1,\ldots,p$,  $E\left\{ \exp(\kappa|\hat{\theta}_{j,s}|)\right\}<D_0 $.
\end{condition}

\begin{condition}\label{a2}
The variance of each component of $(Y,\mathbf{X},Z)$ is not 0, and the absolute value of the correlation between any two variables in $(Y,\mathbf{X},Z)$ is not 1.
\end{condition}

\begin{condition}\label{a3}
There exists two constants $c>0$ and $0<\tau< 1/2$ such that $\min_{j \in \mathcal{M}}\omega_j \geq 2cN^{-\tau}$.
\end{condition}

Condition \ref{a1} assumes a sub-exponential distribution for $\hat{\theta}_{j,s}$, which guarantee the boundedness of the finite-order moment of the variable, this condition is widely used in high-dimensional feature screening (e.g., \citet{r2} and \citet{r1}). In fact, if we assume that each component of $(Y,\mathbf{X},Z)$ follows a sub-Gaussian distribution, we can infer that Condition \ref{a1} holds. Condition \ref{a2}, which requires that the features are non-degenerate and that there is no perfect collinearity between the features, can help simplify the proof. Condition \ref{a3} is the minimum signal condition, which requires that the partial correlation of all important features cannot be too small, that is, the minimum signal cannot be too weak. Under the above conditions, we have the following conclusion.

\begin{proposition}\label{pro1}
Under condition \ref{a1}, we have $ \max_{1\leq j \leq p,1\leq s \leq 9 }\mathrm{var}(\bar{\theta}_{j,s})=O(\frac{1}{N}) $.    
\end{proposition}
By Proposition \ref{pro1}, we see that the maximum variance of $\bar{\theta}_{j,s}$ has the same order $O(\frac{1}{N})$ as the centralized estimator for both fixed $K$ and diverging $K$. While a larger $K$ may lead to an increased variance of $\bar{\theta}_{j,s}$, such a precision loss is insignificant in the sense that the asymptotic order remains unchanged.  

\begin{theorem}\label{b1}
Under conditions \ref{a1} and \ref{a2}, for any $\epsilon>0$, there exists positive constants $c_1,\ldots,c_6$ such that
\begin{align*}
  P( \max_{1\leq j \leq p }|\bar{\omega}_j-\omega_j|\geq \epsilon ) &\leq pKc_1(1-\epsilon^2/c_1)^n, \\
  P( \max_{1\leq j \leq p }|\tilde{\omega}_j-\omega_j|\geq \epsilon ) &\leq pc_2(1-\epsilon^2/c_2)^N, \\
  P( \max_{1\leq j \leq p }|\bar{\omega}^D_j-\omega_j|\geq \epsilon ) & \leq pKc_3(1-\epsilon^2/c_3)^n+p\delta(\epsilon),
\end{align*}
where $\delta(\epsilon)=36K\exp(-c_4n\min(c^2_5\epsilon^2,c_5\epsilon))+36N\exp(-c_6(n-1))$.
\end{theorem}
Theorem \ref{b1} states that the three estimators converge in probability to $\omega_j$, that is, the consistency of the screening utility is satisfied.

\begin{theorem}[Sure Screening Property]\label{b2}
Let $\gamma=cN^{-\tau}$. Under conditions \ref{a1}-\ref{a3}, there exists positive constants $d_1,\ldots,d_6$ such that
\begin{align*}
  P( \mathcal{M}\subseteq \bar{\mathcal{M}} ) &\geq 1-|\mathcal{M}|Kd_1(1-N^{-2\tau}/d_1)^n, \\
  P( \mathcal{M}\subseteq \tilde{\mathcal{M}} ) &\geq 1-|\mathcal{M}|d_2(1-N^{-2\tau}/d_2)^N, \\
  P( \mathcal{M}\subseteq \bar{\mathcal{M}}^D ) &\geq 1-|\mathcal{M}|Kd_3(1-N^{-2\tau}/d_3)^n-|\mathcal{M}|\delta(cN^{-\tau}),
\end{align*}
where $\delta(cN^{-\tau})=36K\exp(-d_4n\min(d^2_5N^{-2\tau},d_5N^{-\tau}))+36N\exp(-d_6(n-1))$, and $|\mathcal{M}|$ represents the number of elements in set $\mathcal{M}$. 
\end{theorem}
Theorem \ref{b2} shows that our methods retain all important features with probability tending to 1. Thus the proposed methods enjoy the sure screening property. If $|\mathcal{M}|=O(e^{N^{v_1}})$ and $K=O(N^{v_2})$ with $v_1,v_2>0$ and $v_1<1-2\tau-v_2$, then the probability lower bound in theorem \ref{b2} approaches one as $N \to \infty$.

\begin{theorem}\label{b3}
Let $\hat{\mathcal{M}}$ denote the estimator of $ \mathcal{M} $ based on  $\bar{\omega_j}$, $\tilde{\omega_j}$ or $\bar{\omega_j}^D$. Under conditions \ref{a1}-\ref{a3}, if $\gamma=cN^{-\tau}$, we have
\begin{equation*}
  P( |\hat{\mathcal{M}}|\leq 2c^{-1}N^{\tau}\sum^{p}_{j=1}\omega_j )\geq 1-o(1).
\end{equation*}
\end{theorem}
Similarly to Theorem \ref{b2}, if we let $p=O(e^{N^{v_1}})$ and $K=O(N^{v_2})$ with $v_1,v_2>0$ and $v_1<1-2\tau-v_2$, then the probability lower bound in Theorem \ref{b3} approaches on. Theorem \ref{b3} implies that if $\sum^{p}_{j=1}\omega_j$ is a polynomial order of $N$, even though $p$ diverges exponentially, the proposed methods can asymptotically control $|\mathcal{M}|$ to be a polynomial order of $N$.

\begin{theorem}[Ranking Consistency]\label{b4}
Let $\hat{\omega}_j$ denote three distributed estimators of $\omega_j$ as above. Under conditions \ref{a1}-\ref{a3}, if $\min_{j \in \mathcal{M}}\omega_j-\max_{j \in \mathcal{M}^c}\omega_j=\varDelta>0$, we have
\begin{equation*}
  P\left( \min_{j \in \mathcal{M}}\hat{\omega}_j>\max_{j \in \mathcal{M}^c}\hat{\omega}_j \right)\geq 1-o(1).
\end{equation*}
\end{theorem}
Theorem \ref{b4} implies that our methods enjoy the ranking consistency
property, that is, our methods can rank an important feature above an unimportant feature with probability approaching one. It is not necessary that $\varDelta$ is a constant. In fact, if $\varDelta \geq O(N^{-\tau})$ and $(p,K)$ satisfies the same conditions as before, the theorem still holds.

\section{Distributed Feature Screening Based on FDR Control}

In Section 2, we have assumed that $\gamma$ is a pre-specified threshold. In Section 3, we prove that our feature screening methods enjoy sure screening property when $\gamma=cN^{-\tau}$ with a constant $c$. However, it is not easy to determine an appropriate value of $\gamma$ in practice. A too large value of $\gamma$ may lead to the omission of some important features, while a small value of $\gamma$ may cause many unimportant features to be falsely selected. Although the constant $c$ can be selected by some data-driven information criteria. But an information criterion strictly depends on a correctly specified model, which is also difficult in real data analysis. Inspired by the work of \citet{r7}, we investigate a distributed feature screening method, which determines the threshold through knockoff features while controlling FDR. 

\subsection{Knockoff Features}
The concept of knockoff variables was first proposed in \citet{r8}. We say $\tilde{\mathbf{X}}={(\tilde{X_1},\ldots,\tilde{X_p})}^{T}$ is a knockoff copy of $\mathbf{X}$ if it satisfies the following conditions.

\begin{condition}\label{a4}
(a) For any $S\subset \{1,\ldots,p\} $, $ (\mathbf{X},\tilde{\mathbf{X}})=_{d}(\mathbf{X},\tilde{\mathbf{X}})_{swap(S)} $;
(b) $\tilde{\mathbf{X}}\perp\!\!\!\!\perp Y|\mathbf{X}$, that is, given $\mathbf{X}$, $Y$ is independent of $\tilde{\mathbf{X}}$.
\end{condition}
In the above condition, "$=_{d}$" means that both sides have the same distribution, and $(\mathbf{ X} ,\tilde{ \mathbf{ X}} )_{ swap(S)} $ means that for all $j \in S$, the position of $\tilde{ \mathbf{ X_j}} $ and $\mathbf{ X_j} $ is swapped. In order to adapt to the partial correlation considered in this paper, we modify the above definition slightly as follows.

\begin{condition}\label{a5}
(a) For any $S\subset \{1,\ldots,p\} $, $ (\mathbf{X},\tilde{\mathbf{X}})=_{d}(\mathbf{X}, \tilde{\mathbf{X}})_{swap(S)} $;
(b) $\tilde{\mathbf{X}}\perp\!\!\!\!\perp (Y,Z)|\mathbf{X}$, that is, given $\mathbf{X}$, $(Y,Z)$ is independent of $\tilde{\mathbf{X}}$.
\end{condition}

Condition \ref{a5}(b) is satisfied if $\tilde{\mathbf{X}}$ is generated without using the information of $(Y,Z)$. But generating the knockoff features that exactly follow condition \ref{a5}(a) requires the information of the distribution of $\mathbf{X}$ to be known in advance, which is usually not available in practice. Thus, we consider constructing approximately second-order knockoff features, keeping the mean and covariance of $(\mathbf{X},\tilde{\mathbf{X}})_{swap(S)}$ the same as that of $(\mathbf{X},\tilde{\mathbf{X}})$. According to the method of \citet{r8}, we use $\mathcal{X} \in R^{n\times p}$ to represent the sample matrix of $\mathbf{X}$ with $n$ samples, and construct $\tilde{\mathcal{X}}$ based on $\mathcal{X}$. Let $\Sigma=\mathcal{X}^T\mathcal{X}$. It follows that
$$ \tilde{\mathcal{X}}=\mathcal{X}(I-\Sigma^{-1}\mathrm{diag}\{\mathbf{s}\})+\tilde{U}C, $$
where $\tilde{U} \in R^{n\times p}$ is an orthogonal matrix,  $\tilde{U}^T\mathcal{X}=\mathbf{0}$, $C^TC=2\mathrm{diag}\{\mathbf{s}\}-\mathrm{diag}\{\mathbf{s}\}\Sigma^{-1}\mathrm{diag}\{\mathbf{s}\} \succeq \mathbf{0}$, $A \succeq \mathbf{0}$ indicates that $A$ is a positive semi-definite matrix, $\mathrm{diag}\{\mathbf{s}\}$ is the diagonal matrix generated by the vector $\mathbf{s}$, and $\mathrm{diag}\{\mathbf{s}\} \succeq \mathbf{0}$ and $ 2\Sigma \succeq \mathrm{diag}\{\mathbf{s}\} $. Then we have
\begin{equation}
  [\mathcal{X}\quad \tilde{\mathcal{X}}]^T[\mathcal{X}\quad \tilde{\mathcal{X}}]= \begin{bmatrix} \Sigma & \Sigma-\mathrm{diag}\{\mathbf{s}\} \\ \Sigma-\mathrm{diag}\{\mathbf{s}\} & \Sigma \end{bmatrix}.\\ 
\end{equation}

\citet{r8} introduced two approaches to construct $\mathrm{diag}\{\mathbf{s}\}$. The first approach is known as the equicorrelated construction, which sets 
$$ s_j=2\lambda_{\min}(\Sigma)\land 1 \quad j=1,\ldots,p, $$ 
where the diagonal elements of $\Sigma$ are all one by scaling the data, and $\lambda_{\min}(\Sigma)$ denotes the smallest eigenvalue of $\Sigma$. Another method, named semidefinite programming, construct it by solving an optimization problem:
$$ \min_{s_j} \sum_{j}|1-s_j|\quad s.t. \quad s_j \geq 0 ,\quad 2\Sigma \succeq \mathrm{diag}\{\mathbf{s}\}. $$
Note that the above process of generating $\tilde{\mathcal{X}}$ requires $n\geq 2p$ to ensure the existence of $\tilde{U}$.

\subsection{Correlation Based on Knockoff Features}

We use $\psi_j$ to measure the partial association between the response $Y$ and the feature $X_j$ based on knockoff features. We define
\begin{equation}
    \psi_j=|\rho_{Y,X_{j}|Z}|-|\rho_{Y,\tilde{X_{j}}|Z}|=\omega(Y,X_j,Z)-\omega(Y,\tilde{X_j},Z),
\end{equation}
where $\omega(Y,\tilde{X_j},Z)$ denotes the $\omega_j$ calculated based on $(Y,\tilde{X_j},Z)$. Obviously, we can get an estimator
\begin{equation} \label{psihat}
    \hat{\psi}_j=\hat{\omega}(\mathcal{Y},\mathcal{X}_j,\mathcal{Z})-\hat{\omega}(\mathcal{Y},\tilde{\mathcal{X}}_j,\mathcal{Z}),
\end{equation}
where $(\mathcal{Y},\mathcal{X},\mathcal{Z})$ respectively denote the sample matrix or vector corresponding to $(Y,X,Z)$, and $\mathcal{X}_j$ denotes the $j$th column of $\mathcal{X}$. In Section 2, we propose three distributed estimators of $\omega_j$, which are uniformly denoted as $\hat{ \omega}_j$. We use $\hat{ \psi}_j$ to denote the three estimators of $\psi_j$, where subscripts will be marked to indicate the specific method if necessary. It is noted that $\tilde{X_j}$ is a "pseudo" feature generated based on $X_j$. If $X_j$ is an important feature, the value of $|\hat{\psi}_j|$ should be large, otherwise it should be close to 0, and the probability of $\hat{\psi}_j$ being positive or negative is the same. We have the following conclusion.

\begin{lemma}\label{knockoff-lemma}
(a)For all $j \in \mathcal{M}^c$, $\psi_j=0$;
(b)Given $|\hat{\psi}|=(|\hat{\psi}_1|,\ldots,|\hat{\psi}_p|)^T$, let $\mathcal{M}^c=\{ 
 j_1,\ldots,j_r \}$, $I_{j_1},\ldots,I_{j_r}$ follow iid 
 $Bernoulli(0.5)$, where $I_{j_k}=I( \hat{\psi}_{j_k} > 0 )$, $I(\cdot)$ denotes characteristic function.
\end{lemma}
Note that when $\hat{\psi}_j=0$, the corresponding feature will not be selected into the model. The result of Lemma \ref{knockoff-lemma} provides us with a method to control FDR. In feature screening, for a certain threshold $t>0$, FDR is defined as
$$ FDR(t)=E[FDP(t)]=E\left[ \frac{\#\{ j \in \mathcal{M}^c:\hat{\psi}_j \geq t \}}{\#\{ j :\hat{\psi}_j \geq t \}} \right], $$
where $\#\{ \cdot \}$ denotes the number of elements in the set. Since $\mathcal{M}^c$ is unknown, $\#\{ j \in \mathcal{M}^c:\hat{\psi}_j \geq t \}$ is unknown in practice. According to Lemma \ref{knockoff-lemma}, we have
$$ \#\{ j \in \mathcal{M}^c:\hat{\psi}_j \geq t \} \approx \#\{ j \in \mathcal{M}^c:\hat{\psi}_j \leq -t \} \leq \#\{ j :\hat{\psi}_j \leq -t \}. $$
To this end, we define an approximate conservative estimator of FDP(t) as
$$ \widehat{FDP}(t)=\frac{\#\{ j :\hat{\psi}_j \leq -t \}}{\#\{ j :\hat{\psi}_j \geq t \}}. $$

By selecting an appropriate threshold $t$, we can control FDP indirectly, and thus control FDR under a pre-given level $\alpha$. According to the method in \citet{r8}, we get the threshold through the following equation:
\begin{equation} \label{talpha}
    T_{\alpha}=\min \left\{ t \in \mathcal{W}: \frac{1+\#\{ j :\hat{\psi}_j \leq -t \}}{\#\{ j :\hat{\psi}_j \geq t \}} \leq \alpha \right\},
\end{equation}
where $\mathcal{W}=\{ |\hat{\psi}_j|:1\leq j \leq p \}/\{ 0 \}$. Note that we added a unit to the numerator of the expression, which can control FDP to a greater extent. 
Therefore, according to the $T_{\alpha}$ in equation \eqref{talpha}, we get a new estimator of $ \mathcal{M}$ as
\begin{equation}
    \hat{\mathcal{M}}(T_{\alpha})=\{ j :\hat{\psi}_j \geq T_{\alpha},1 \leq j \leq p \}.
\end{equation}

\subsection{FDR Control Based on Knockoff Features}
Because generating knockoff features requires that the relationship between data sample size and dimension satisfies $n \geq 2p$, the above knockoff method cannot be directly used for ultrahigh-dimensional feature screening. We consider a two-step feature screening strategy as follows. First, the data are divided into two parts on each machine. For the first part, one of the feature screening methods in Section 2 is employed to select $d$ features such that $2d<n_2$. Then, for the second part of the data on each machine, the selected $d$ features are applied to the knockoff method introduced in Section 4.1 to generate corresponding knockoff features. Second, based on the knockoff features from step 1, we further filter the features by implementing the method in Section 4.2. We summarize this procedure in Algorithm \ref{algorithm1}.

\begin{algorithm}
 \caption{ Distributed feature screening based on Knockoff}
 \label{algorithm1}
 \begin{algorithmic}[1]
   \Require sample matrix $(\mathcal{Y},\mathcal{X},\mathcal{Z})$, FDR level $\alpha$, number of machines $K$, $(n_1,n_2,d)$ that satisfies  $n_1+n_2=n=N/K, n_2>2d$;
   \State Partition equally full set $ \mathcal{D} $ into $K$ subsets and distribute them to $K$ machines, denoted as $\mathcal{D}_{(k)}=(\mathcal{Y}_{(k)},\mathcal{X}_{(k)},\mathcal{Z}_{(k)})$, on each machine, partition $\mathcal{D}_{(k)}$ into $\mathcal{D}_{(k)}^{(1)}\in R^{n_1 \times (p+2)}$ and $\mathcal{D}_{(k)}^{(2)}\in R^{n_2 \times (p+2)}$;
   \State Compute $\hat{\omega}_j$ based on $\{ \mathcal{D}_{(k)}^{(1)}, k=1,\ldots,K\}$;
   \State Perform the first feature screening to get $ \hat{\mathcal{M}}_1=\{ j:\hat{\omega}_j $ is among the largest $d$ $ \} $; 
   \State Based on $\{ \mathcal{D}_{(k)}^{(2)}, k=1,\ldots,K\}$, for all $j \in \hat{\mathcal{M}}_1$, select the corresponding data columns $\mathcal{X}_{(k)j}^{(2)}$ to form the matrix $\mathcal{X}_{(k)(\hat{\mathcal{M}}_1)}^{(2)}$, and generate the knockoff features $\tilde{\mathcal{X}}_{(k)(\hat{\mathcal{M}}_1)}^{(2)}$ according to the method in section 4.1;
   \State Based on $\{ (\mathcal{Y}_{(k)}^{(2)},\mathcal{X}_{(k)(\hat{\mathcal{M}}_1)}^{(2)},\mathcal{Z}_{(k)}^{(2)}),k=1,\ldots,K\}$ and $\{ (\mathcal{Y}_{(k)}^{(2)},\tilde{\mathcal{X}}_{(k)(\hat{\mathcal{M}}_1)}^{(2)},\mathcal{Z}_{(k)}^{(2)}),k=1,\ldots,K\}$, compute $\hat{\omega}(\mathcal{Y},\mathcal{X}_j,\mathcal{Z})$ and $\hat{\omega}(\mathcal{Y},\tilde{\mathcal{X}}_j,\mathcal{Z}),~ j \in \hat{\mathcal{M}}_1$;
   \State Based on equation \eqref{psihat}, compute $\hat{\psi}_j$;
   \State Based on equation \eqref{talpha}, choose the threshold $T_{\alpha}$;
   \Ensure Obtain $\hat{\mathcal{M}}(T_{\alpha})=\{ j :\hat{\psi}_j \geq T_{\alpha},j \in \hat{\mathcal{M}}_1 \}$.
 \end{algorithmic}
\end{algorithm}

To ensure that the first step satisfies the sure screening property, we choose as large $d$ as possible when implementing the Algorithm \ref{algorithm1}, so that $2d<n_2$ holds. At the same time, to generate better knockoff features, we also choose $n_2$ as large as possible. The following theorems provide some theoretical properties about the Algorithm \ref{algorithm1}.

\begin{theorem}\label{knock-th}
Let $\mathcal{E}$ denote the event $\{ \mathcal{M} \subseteq \hat{\mathcal{M}}_1 \}$. If the knockoff features satisfy condition \ref{a5}, then for any $\alpha \in [0,1]$, we have
$$ FDR=E\left[ \frac{\#\{j: j \in \mathcal{M}^c \cap \hat{\mathcal{M}}(T_{\alpha}) \}}{\#\{ j :j \in  \hat{\mathcal{M}}(T_{\alpha}) \}\lor 1 }\bigg|\mathcal{E} \right]\leq \alpha. $$
\end{theorem}
Theorem \ref{knock-th} shows that our two-step method can control FDR under the given level $\alpha$. Since the first step of the algorithm satisfies sure screening property, that is, $\mathcal{E}$ occurs with probability tending to 1, the FDR can be controlled even without conditioning on $\mathcal{E}$. This is verified by our simulation results.

\begin{theorem}\label{knock-th2}
Under the conditions of Theorem \ref{knock-th} and conditions \ref{a1} and \ref{a2}, if there exist constants $c>0, 0<\tau<1/2$ such that $\min_{j \in \mathcal{M}}\psi_j \geq 4c(Kn_2)^{-\tau} $, then (a) if $\alpha \geq 1/|\mathcal{M}|$, we have $P( \mathcal{M} \subseteq \hat{\mathcal{M}}(T_{\alpha}) | \mathcal{E} ) \geq 1-o(1) $; (b) if $\alpha < 1/|\mathcal{M}|$, we have $P( \{\mathcal{M} \subseteq \hat{\mathcal{M}}(T_{\alpha})\} \cup \{ \hat{\mathcal{M}}(T_{\alpha})=\emptyset \} | \mathcal{E} ) \geq 1-o(1) .$
\end{theorem}
Theorem \ref{knock-th2} shows that under the condition $\alpha \geq 1/|\mathcal{M}|$, the screening methods based on knockoff enjoy sure screening property and control the FDR simultaneously with high probability. However, when $\alpha < 1/|\mathcal{M}|$, the sure screening property cannot be guaranteed, which is verified by our simulation results as well.

\section{Numerical Simulation}
In this section, numerical simulations are carried out to show the performance of the proposed methods in finite samples. We use SAPS, ACPS and JDPS to represent the three different distributed feature screening methods in Section 2; and use SAPS-Kn, ACPS-Kn and JDPS-Kn to represent the three distributed feature screening methods based on knockoff features for FDR control in Section 4. In each example, we repeat the procedure 200 times and set $\Sigma=(\sigma_{ij})_{p \times p}$, $\sigma_{ij}=0.5^{|i-j|}$.

\begin{table}[htbp]
    \centering
    \caption{Simulation results of Model (a) in Example 1 under the setting of $c=0.02$. }\label{tab1}
   \resizebox{.8\columnwidth}{!}{
    \begin{tabular*}{\hsize}{@{}@{\extracolsep{\fill}}cccccccccccc@{}}
    \hline
        $Z$ & ($N$,$p$) & $K$ & method & 5\% & 50\% & 95\% & AUC & SSR & PSR & FDR & Time \\ \hline
        $X_1$ & 10000 & 1 & Global & 7 & 10 & 100 & 0.9987 & 1 & 1 & 0.99 & 8.57 \\ 
        ~ & 3000 & 20 & SAPS & 7 & 10 & 109 & 0.9986 & 1 & 1 & 0.99 & 0.85 \\ 
        ~ & ~ & ~ & ACPS & 7 & 10 & 100 & 0.9987 & 1 & 1 & 0.99 & 3.21 \\ 
        ~ & ~ & ~ & JDPS & 7 & 10 & 109 & 0.9986 & 1 & 1 & 0.99 & 14.46 \\ 
        ~ & ~ & 50 & SAPS & 7 & 11 & 121 & 0.9985 & 1 & 1 & 0.99 & 0.55 \\ 
        ~ & ~ & ~ & ACPS & 7 & 10 & 100 & 0.9987 & 1 & 1 & 0.99 & 1.88 \\ 
        ~ & ~ & ~ & JDPS & 7 & 10 & 126 & 0.9985 & 1 & 1 & 0.99 & 24.00 \\ 
        ~ & 20000 & 1 & Global & 7 & 7 & 9 & 0.9999 & 1 & 1 & 0.99 & 9.89 \\ 
        ~ & 3000 & 20 & SAPS & 7 & 7 & 9 & 0.9999 & 1 & 1 & 0.99 & 0.82 \\ 
        ~ & ~ & ~ & ACPS & 7 & 7 & 9 & 0.9999 & 1 & 1 & 0.99 & 3.10 \\ 
        ~ & ~ & ~ & JDPS & 7 & 7 & 9 & 0.9999 & 1 & 1 & 0.99 & 7.39 \\ 
        ~ & ~ & 50 & SAPS & 7 & 7 & 9 & 0.9999 & 1 & 1 & 0.99 & 0.50 \\ 
        ~ & ~ & ~ & ACPS & 7 & 7 & 9 & 0.9999 & 1 & 1 & 0.99 & 2.09 \\ 
        ~ & ~ & ~ & JDPS & 7 & 7 & 9 & 0.9999 & 1 & 1 & 0.99 & 13.63 \\ 
        ~ & 20000 & 1 & Global & 7 & 7 & 9 & 0.9999 & 1 & 1 & 0.99 & 19.01 \\ 
        ~ & 6000 & 20 & SAPS & 7 & 7 & 10 & 0.9999 & 1 & 1 & 0.99 & 1.45 \\ 
        ~ & ~ & ~ & ACPS & 7 & 7 & 9 & 0.9999 & 1 & 1 & 0.99 & 4.74 \\ 
        ~ & ~ & ~ & JDPS & 7 & 7 & 10 & 0.9999 & 1 & 1 & 0.99 & 15.82 \\ 
        ~ & ~ & 50 & SAPS & 7 & 7 & 10 & 0.9999 & 1 & 1 & 0.99 & 1.71 \\ 
        ~ & ~ & ~ & ACPS & 7 & 7 & 9 & 0.9999 & 1 & 1 & 0.99 & 3.07 \\ 
        ~ & ~ & ~ & JDPS & 7 & 7 & 10 & 0.9999 & 1 & 1 & 0.99 & 19.08 \\ \hline
        $X_2$ & 10000 & 1 & Global & 14 & 315 & 2457 & 0.9721 & 0.77 & 0.97 & 0.99 & 8.53 \\ 
        ~ & 3000 & 20 & SAPS & 14 & 307 & 2506 & 0.9717 & 0.75 & 0.97 & 0.99 & 0.81 \\ 
        ~ & ~ & ~ & ACPS & 14 & 315 & 2457 & 0.9721 & 0.77 & 0.97 & 0.99 & 2.43 \\ 
        ~ & ~ & ~ & JDPS & 14 & 308 & 2498 & 0.9717 & 0.75 & 0.97 & 0.99 & 14.29 \\ 
        ~ & ~ & 50 & SAPS & 17 & 307 & 2504 & 0.9717 & 0.77 & 0.97 & 0.99 & 0.42 \\ 
        ~ & ~ & ~ & ACPS & 14 & 315 & 2457 & 0.9721 & 0.77 & 0.97 & 0.99 & 1.65 \\ 
        ~ & ~ & ~ & JDPS & 17 & 307 & 2464 & 0.9717 & 0.77 & 0.97 & 0.99 & 18.92 \\ 
        ~ & 20000 & 1 & Global & 8 & 42 & 1137 & 0.9913 & 0.98 & 0.99 & 0.99 & 10.12 \\ 
        ~ & 3000 & 20 & SAPS & 8 & 42 & 1131 & 0.9913 & 0.98 & 0.99 & 0.99 & 0.76 \\ 
        ~ & ~ & ~ & ACPS & 8 & 42 & 1137 & 0.9913 & 0.98 & 0.99 & 0.99 & 2.94 \\ 
        ~ & ~ & ~ & JDPS & 8 & 42 & 1128 & 0.9913 & 0.98 & 0.99 & 0.99 & 8.07 \\ 
        ~ & ~ & 50 & SAPS & 8 & 43 & 1109 & 0.9912 & 0.98 & 0.99 & 0.99 & 0.64 \\ 
        ~ & ~ & ~ & ACPS & 8 & 42 & 1137 & 0.9913 & 0.98 & 0.99 & 0.99 & 2.34 \\ 
        ~ & ~ & ~ & JDPS & 8 & 43 & 1091 & 0.9912 & 0.98 & 0.99 & 0.99 & 11.68 \\ 
        ~ & 20000 & 1 & Global & 9 & 87 & 2973 & 0.9889 & 0.92 & 0.99 & 0.99 & 19.56 \\ 
        ~ & 6000 & 20 & SAPS & 9 & 83 & 3029 & 0.9889 & 0.92 & 0.99 & 0.99 & 1.81 \\ 
        ~ & ~ & ~ & ACPS & 9 & 87 & 2973 & 0.9889 & 0.92 & 0.99 & 0.99 & 4.59 \\ 
        ~ & ~ & ~ & JDPS & 9 & 83 & 3012 & 0.9889 & 0.92 & 0.99 & 0.99 & 15.75 \\ 
        ~ & ~ & 50 & SAPS & 9 & 98 & 2920 & 0.9887 & 0.92 & 0.99 & 0.99 & 1.50 \\ 
        ~ & ~ & ~ & ACPS & 9 & 87 & 2973 & 0.9889 & 0.92 & 0.99 & 0.99 & 3.22 \\ 
        ~ & ~ & ~ & JDPS & 9 & 98 & 2939 & 0.9887 & 0.92 & 0.99 & 0.99 & 19.66 \\ \hline
    \end{tabular*}
    }
\end{table}

\begin{table}[htbp]
    \centering
    \caption{Simulation results of Model (b) in Example 1 under the setting of $c=0.03$. }\label{tab2}
    \resizebox{.8\columnwidth}{!}{
    \begin{tabular*}{\hsize}{@{}@{\extracolsep{\fill}}cccccccccccc@{}}
    \hline
        $Z$ & ($N$,$p$) & $K$ & method & 5\% & 50\% & 95\% & AUC & SSR & PSR & FDR & Time \\ \hline
        $X_5$ & 10000 & 1 & Global & 4 & 4 & 4 & 1 & 1 & 1 & 0.99 & 4.48 \\ 
        ~ & 3000 & 20 & SAPS & 4 & 4 & 4 & 1 & 1 & 1 & 0.99 & 0.48 \\ 
        ~ & ~ & ~ & ACPS & 4 & 4 & 4 & 1 & 1 & 1 & 0.99 & 1.82 \\ 
        ~ & ~ & ~ & JDPS & 4 & 4 & 4 & 1 & 1 & 1 & 0.99 & 12.70 \\ 
        ~ & ~ & 50 & SAPS & 4 & 4 & 4 & 1 & 1 & 1 & 0.99 & 0.40 \\ 
        ~ & ~ & ~ & ACPS & 4 & 4 & 4 & 1 & 1 & 1 & 0.99 & 1.31 \\ 
        ~ & ~ & ~ & JDPS & 4 & 4 & 4 & 1 & 1 & 1 & 0.99 & 16.87 \\ 
        ~ & 20000 & 1 & Global & 4 & 4 & 4 & 1 & 1 & 1 & 0.99 & 9.53 \\ 
        ~ & 3000 & 20 & SAPS & 4 & 4 & 4 & 1 & 1 & 1 & 0.99 & 0.70 \\ 
        ~ & ~ & ~ & ACPS & 4 & 4 & 4 & 1 & 1 & 1 & 0.99 & 2.32 \\ 
        ~ & ~ & ~ & JDPS & 4 & 4 & 4 & 1 & 1 & 1 & 0.99 & 8.42 \\ 
        ~ & ~ & 50 & SAPS & 4 & 4 & 4 & 1 & 1 & 1 & 0.99 & 0.51 \\ 
        ~ & ~ & ~ & ACPS & 4 & 4 & 4 & 1 & 1 & 1 & 0.99 & 1.86 \\ 
        ~ & ~ & ~ & JDPS & 4 & 4 & 4 & 1 & 1 & 1 & 0.99 & 9.37 \\ 
        ~ & 20000 & 1 & Global & 4 & 4 & 4 & 1 & 1 & 1 & 0.99 & 18.47 \\ 
        ~ & 6000 & 20 & SAPS & 4 & 4 & 4 & 1 & 1 & 1 & 0.99 & 1.41 \\ 
        ~ & ~ & ~ & ACPS & 4 & 4 & 4 & 1 & 1 & 1 & 0.99 & 4.91 \\ 
        ~ & ~ & ~ & JDPS & 4 & 4 & 4 & 1 & 1 & 1 & 0.99 & 15.19 \\ 
        ~ & ~ & 50 & SAPS & 4 & 4 & 4 & 1 & 1 & 1 & 0.99 & 0.91 \\ 
        ~ & ~ & ~ & ACPS & 4 & 4 & 4 & 1 & 1 & 1 & 0.99 & 3.46 \\ 
        ~ & ~ & ~ & JDPS & 4 & 4 & 4 & 1 & 1 & 1 & 0.99 & 9.94 \\ \hline
        $X_6$ & 10000 & 1 & Global & 5 & 5 & 9 & 0.99 & 1 & 1 & 0.99 & 5.13 \\ 
        ~ & 3000 & 20 & SAPS & 5 & 5 & 11 & 0.99 & 1 & 1 & 0.99 & 0.41 \\ 
        ~ & ~ & ~ & ACPS & 5 & 5 & 9 & 0.99 & 1 & 1 & 0.99 & 1.76 \\ 
        ~ & ~ & ~ & JDPS & 5 & 5 & 11 & 0.99 & 1 & 1 & 0.99 & 12.48 \\ 
        ~ & ~ & 50 & SAPS & 5 & 5 & 9 & 0.99 & 1 & 1 & 0.99 & 0.25 \\ 
        ~ & ~ & ~ & ACPS & 5 & 5 & 9 & 0.99 & 1 & 1 & 0.99 & 1.15 \\ 
        ~ & ~ & ~ & JDPS & 5 & 5 & 10 & 0.99 & 1 & 1 & 0.99 & 12.30 \\ 
        ~ & 20000 & 1 & Global & 5 & 5 & 5 & 1 & 1 & 1 & 0.99 & 10.96 \\ 
        ~ & 3000 & 20 & SAPS & 5 & 5 & 5 & 1 & 1 & 1 & 0.99 & 0.79 \\ 
        ~ & ~ & ~ & ACPS & 5 & 5 & 5 & 1 & 1 & 1 & 0.99 & 2.82 \\ 
        ~ & ~ & ~ & JDPS & 5 & 5 & 5 & 1 & 1 & 1 & 0.99 & 8.55 \\ 
        ~ & ~ & 50 & SAPS & 5 & 5 & 5 & 1 & 1 & 1 & 0.99 & 0.76 \\ 
        ~ & ~ & ~ & ACPS & 5 & 5 & 5 & 1 & 1 & 1 & 0.99 & 2.13 \\ 
        ~ & ~ & ~ & JDPS & 5 & 5 & 5 & 1 & 1 & 1 & 0.99 & 9.05 \\ 
        ~ & 20000 & 1 & Global & 5 & 5 & 5 & 1 & 1 & 1 & 0.99 & 18.53 \\ 
        ~ & 6000 & 20 & SAPS & 5 & 5 & 5 & 1 & 1 & 1 & 0.99 & 1.31 \\ 
        ~ & ~ & ~ & ACPS & 5 & 5 & 5 & 1 & 1 & 1 & 0.99 & 4.76 \\ 
        ~ & ~ & ~ & JDPS & 5 & 5 & 5 & 1 & 1 & 1 & 0.99 & 14.90 \\ 
        ~ & ~ & 50 & SAPS & 5 & 5 & 5 & 1 & 1 & 1 & 0.99 & 0.89 \\ 
        ~ & ~ & ~ & ACPS & 5 & 5 & 5 & 1 & 1 & 1 & 0.99 & 3.46 \\ 
        ~ & ~ & ~ & JDPS & 5 & 5 & 5 & 1 & 1 & 1 & 0.99 & 9.77 \\ \hline
    \end{tabular*}
    }
\end{table}

We evaluate the performance of proposed methods in terms of successful screening rate (SSR), positive selection rate (PSR), false discovery rate (FDR) and AUC measure, which are calculated as follows:
$$ SSR=\frac{1}{T}\sum_{t=1}^{T}I\{ \mathcal{M} \subseteq \hat{\mathcal{M}}(t) \} ,\quad PSR=\frac{| \mathcal{M} \cap \hat{\mathcal{M}}(t) |}{| \mathcal{M} |},\quad FDR=\frac{| \hat{\mathcal{M}}(t)-\mathcal{M} |}{|  \hat{\mathcal{M}}(t) |} ,$$
$$ AUC=1-   \frac{1}{|\mathcal{M}||\mathcal{M}^c|} \sum_{i \in \mathcal{M}} \sum_{j \in \mathcal{M}^c} \left[I\{ \hat{\omega}_i < \hat{\omega}_j \} + 0.5I\{ \hat{\omega}_i = \hat{\omega}_j \} \right], $$
where $T$ denotes the number of repeated simulations, $\hat{\mathcal{M}}(t)$ denotes the index set of the features retained after screening at the $t$-th repetition based on a hard threshold. In other words, we retain the features whose $\hat{\omega}_j$ are among the largest $\lfloor n/\log(n) \rfloor$. The AUC was used in \citet{r10}, and has been widely used to evaluate the ranking quality (e.g., \citet{r1}). In each repetition, we record the minimum model size (MS) that contains all important features. We report the $5\%,50\%$ and $95\%$ quantiles of the minimum model size, and the means of PSR, FDR, AUC and the CPU time (in seconds).

\subsection{Example 1}
We generate $N$ independent copies $\left\{(Y_i,\mathbf{X}_i,Z_i) \right\}^{N}_{i=1}$, where $\mathbf{X}_i=(X_{i1},\ldots,X_{ip})^{T} \sim N(0,\Sigma)$. The corresponding response $Y$ is generated based on the following models.
\begin{align*}
\text{(a)}~~ Y&=c(X_1+X_3+X_4+X_5+X_6+X_7+X_8+X_9)+\epsilon , \\
\text{(b)}~~ Y&=c(2X_1+3X_2+1.5X_3+2X_4+2\sin(2\pi X_5))+\epsilon,
\end{align*}
where Model (a) is adapted from Model 7 of \citet{r11}, and Model (b) is used to show the performance of the screening methods when a nonlinear term exists. The error terms in the two models satisfy $\epsilon \sim N(0,1)$. In Model (a), we consider two cases where the conditional feature is taken as $Z=X_1$ (important feature) and $Z=X_2$ (unimportant feature). In Model (b), we consider two cases where $Z=X_5$ and $Z=X_6$.

Table \ref{tab1} shows the simulation results in the case of $c=0.02$ in Model (a). From the quantiles of the minimum model size and AUC in Table \ref{tab1}, it can be seen that the performance of the three distributed feature screening methods improves with the increase of the sample size, but decreases with the increase of the dimension. Besides, the selection of conditional feature $Z$ has a greater impact on the results. When $Z$ is set as an important feature, the overall performance is better. Meanwhile, when $Z$ is set as an unimportant feature, the performance will get worse, but even so, the performance of the distributed algorithm is still very close to the results based on full data. In general, the ACPS method outperforms other methods, and the JDPS and SAPS methods perform slightly worse. In addition, in terms of the SSR, PSR and FDR, although the hard threshold-based screening methods may lead to the model with a high probability of selecting all important features, it also introduces a large number of noisy features.
Table \ref{tab2} shows the simulation results in the case of $c=0.03$ in Model (b). A similar observation can be drawn, 

In addition, as shown in the last column of Tables \ref{tab1} and \ref{tab2}, in most cases, the distributed methods reduce the computing time, especially for the SAPS and ACPS methods. However, in some cases, the running time of JDPS method does not show a significant reduction. The reason is probably that the JDPS method is proposed based on jackknife, which is a time-consuming calculation method when the sample size is large.

\subsection{Example 2}

\begin{sidewaystable}[htbp]
    \centering
    \caption{Simulation results of Model (a) in Example 2 under the setting of $(N,p,c,Z)=(10000,5000, 0.725,X_1)$. } \label{tab3}
   \resizebox{.8\columnwidth}{!}{
    \begin{tabular}{cccccccccccccc}
    \hline
        $K$ & $\alpha$ & method & $\mathcal{P}_3$ & $\mathcal{P}_4$ & $\mathcal{P}_5$ & $\mathcal{P}_6$ & $\mathcal{P}_7$ & $\mathcal{P}_8$ & $\mathcal{P}_9$ & MS & SSR & FDR & Time \\ \hline
        20 & 0.1 & SAPS-Kn & 0.37 & 0.26 & 0.22 & 0.24 & 0.24 & 0.24 & 0.25 & 2 & 0.14 & 0.049 & 5.05 \\ 
        ~ & ~ & ACPS-Kn & 0.40 & 0.22 & 0.25 & 0.23 & 0.25 & 0.22 & 0.23 & 2 & 0.13 & 0.049 & 6.12 \\ 
        ~ & ~ & JDPS-Kn & 0.36 & 0.27 & 0.23 & 0.25 & 0.19 & 0.23 & 0.24 & 2 & 0.13 & 0.048 & 16.40 \\ 
        ~ & 0.2 & SAPS-Kn & 1 & 1 & 1 & 1 & 1 & 1 & 1 & 9 & 1 & 0.167 & 4.57 \\ 
        ~ & ~ & ACPS-Kn & 1 & 1 & 1 & 1 & 1 & 1 & 1 & 9 & 1 & 0.166 & 5.90 \\ 
        ~ & ~ & JDPS-Kn & 1 & 1 & 1 & 1 & 1 & 1 & 1 & 9 & 1 & 0.176 & 16.08 \\ 
        ~ & 0.3 & SAPS-Kn & 1 & 1 & 1 & 1 & 1 & 1 & 1 & 11 & 1 & 0.273 & 4.42 \\ 
        ~ & ~ & ACPS-Kn & 1 & 1 & 1 & 1 & 1 & 1 & 1 & 11 & 1 & 0.275 & 5.62 \\ 
        ~ & ~ & JDPS-Kn & 1 & 1 & 1 & 1 & 1 & 1 & 1 & 11 & 1 & 0.279 & 15.84 \\ \hline
        50 & 0.1 & SAPS-Kn & 0.37 & 0.32 & 0.29 & 0.26 & 0.25 & 0.25 & 0.25 & 3 & 0.16 & 0.054 & 0.65 \\ 
        ~ & ~ & ACPS-Kn & 0.31 & 0.21 & 0.24 & 0.22 & 0.20 & 0.20 & 0.24 & 2 & 0.10 & 0.035 & 1.66 \\ 
        ~ & ~ & JDPS-Kn & 0.40 & 0.27 & 0.25 & 0.23 & 0.26 & 0.23 & 0.26 & 2 & 0.15 & 0.050 & 1.61 \\ 
        ~ & 0.2 & SAPS-Kn & 1 & 1 & 1 & 1 & 1 & 1 & 1 & 9 & 1 & 0.166 & 0.61 \\ 
        ~ & ~ & ACPS-Kn & 1 & 1 & 1 & 1 & 1 & 1 & 1 & 9 & 1 & 0.146 & 1.60 \\ 
        ~ & ~ & JDPS-Kn & 1 & 1 & 1 & 1 & 1 & 1 & 1 & 9 & 1 & 0.190 & 1.59 \\ 
        ~ & 0.3 & SAPS-Kn & 1 & 1 & 1 & 1 & 1 & 1 & 1 & 11 & 1 & 0.282 & 0.63 \\ 
        ~ & ~ & ACPS-Kn & 1 & 1 & 1 & 1 & 1 & 1 & 1 & 11 & 1 & 0.275 & 1.64 \\ 
        ~ & ~ & JDPS-Kn & 1 & 1 & 1 & 1 & 1 & 1 & 1 & 11 & 1 & 0.285 & 1.64 \\ \hline
    \end{tabular}
    }
\end{sidewaystable}

\begin{sidewaystable}[htbp]
    \centering
    \caption{Simulation results of Model (a) in Example 2 under the setting of $(N,p,c,Z)=(10000,5000, 0.725,X_2)$. }  \label{tab4}
   \resizebox{.8\columnwidth}{!}{
\begin{tabular}{ccccccccccccccc}
    \hline
        $K$ & $\alpha$ & method & $\mathcal{P}_1$ & $\mathcal{P}_3$ & $\mathcal{P}_4$ & $\mathcal{P}_5$ & $\mathcal{P}_6$ & $\mathcal{P}_7$ & $\mathcal{P}_8$ & $\mathcal{P}_9$ & MS & SSR & FDR & Time \\ \hline
        20 & 0.1 & SAPS-Kn & 0.22 & 0.22 & 0.44 & 0.34 & 0.35 & 0.33 & 0.32 & 0.33 & 3 & 0.22 & 0.056 & 4.41 \\ 
        ~ & ~ & ACPS-Kn & 0.21 & 0.21 & 0.41 & 0.35 & 0.31 & 0.32 & 0.33 & 0.32 & 3 & 0.21 & 0.061 & 5.62 \\ 
        ~ & ~ & JDPS-Kn & 0.29 & 0.29 & 0.47 & 0.42 & 0.42 & 0.37 & 0.36 & 0.42 & 4 & 0.29 & 0.075 & 15.88 \\ 
        ~ & 0.2 & SAPS-Kn & 1 & 1 & 1 & 1 & 1 & 1 & 1 & 1 & 10 & 1 & 0.159 & 4.43 \\ 
        ~ & ~ & ACPS-Kn & 1 & 1 & 1 & 1 & 1 & 1 & 1 & 1 & 10 & 1 & 0.158 & 5.61 \\ 
        ~ & ~ & JDPS-Kn & 1 & 1 & 1 & 1 & 1 & 1 & 1 & 1 & 10 & 1 & 0.176 & 15.94 \\ 
        ~ & 0.3 & SAPS-Kn & 1 & 1 & 1 & 1 & 1 & 1 & 1 & 1 & 12 & 1 & 0.251 & 4.40 \\ 
        ~ & ~ & ACPS-Kn & 1 & 1 & 1 & 1 & 1 & 1 & 1 & 1 & 12 & 1 & 0.248 & 5.66 \\ 
        ~ & ~ & JDPS-Kn & 1 & 1 & 1 & 1 & 1 & 1 & 1 & 1 & 13 & 1 & 0.289 & 15.89 \\ \hline
        50 & 0.1 & SAPS-Kn & 0.26 & 0.26 & 0.44 & 0.40 & 0.36 & 0.40 & 0.34 & 0.35 & 4 & 0.26 & 0.065 & 0.65 \\ 
        ~ & ~ & ACPS-Kn & 0.26 & 0.26 & 0.42 & 0.42 & 0.35 & 0.36 & 0.41 & 0.35 & 4 & 0.26 & 0.065 & 1.61 \\ 
        ~ & ~ & JDPS-Kn & 0.29 & 0.29 & 0.43 & 0.41 & 0.43 & 0.41 & 0.40 & 0.36 & 4 & 0.29 & 0.077 & 1.61 \\ 
        ~ & 0.2 & SAPS-Kn & 1 & 1 & 1 & 1 & 1 & 1 & 1 & 1 & 10 & 1 & 0.177 & 0.62 \\ 
        ~ & ~ & ACPS-Kn & 1 & 1 & 1 & 1 & 1 & 1 & 1 & 1 & 10 & 1 & 0.175 & 1.57 \\ 
        ~ & ~ & JDPS-Kn & 1 & 1 & 1 & 1 & 1 & 1 & 1 & 1 & 11 & 1 & 0.192 & 1.54 \\ 
        ~ & 0.3 & SAPS-Kn & 1 & 1 & 1 & 1 & 1 & 1 & 1 & 1 & 12 & 1 & 0.280 & 0.62 \\ 
        ~ & ~ & ACPS-Kn & 1 & 1 & 1 & 1 & 1 & 1 & 1 & 1 & 12 & 1 & 0.267 & 1.58 \\ 
        ~ & ~ & JDPS-Kn & 1 & 1 & 1 & 1 & 1 & 1 & 1 & 1 & 13 & 1 & 0.297 & 1.52 \\ \hline
    \end{tabular}
     }
\end{sidewaystable}

\begin{table}[htbp]
    \centering
    \caption{Simulation results of Model (b) in Example 2 under the setting of $(N,p,c)=(10000,10000, 0.499)$. }  \label{tab5}
   \resizebox{.8\columnwidth}{!}{
\begin{tabular}{ccccccccccccc}
    \hline
        $Z$ & $K$ & $\alpha$ & method & $\mathcal{P}_1$ & $\mathcal{P}_2$ & $\mathcal{P}_3$ & $\mathcal{P}_4$ & $\mathcal{P}_5$ & MS & SSR & FDR & Time \\ \hline
        $X_5$ & 20 & 0.1 & SAPS-Kn & 0.01 & 1 & 0.01 & 0.01 & - & 1 & 0.01 & 0.003 & 3.66 \\ 
        ~ & ~ & ~ & ACPS-Kn & 0.01 & 1 & 0.01 & 0.01 & - & 1 & 0.01 & 0.004 & 4.94 \\ 
        ~ & ~ & ~ & JDPS-Kn & 0.02 & 1 & 0.02 & 0.02 & - & 1 & 0.02 & 0.013 & 25.14 \\ 
        ~ & ~ & 0.2 & SAPS-Kn & 0.48 & 1 & 0.48 & 0.48 & - & 4 & 0.48 & 0.152 & 3.20 \\ 
        ~ & ~ & ~ & ACPS-Kn & 0.48 & 1 & 0.48 & 0.48 & - & 4 & 0.48 & 0.161 & 4.66 \\ 
        ~ & ~ & ~ & JDPS-Kn & 0.58 & 1 & 0.58 & 0.58 & - & 4 & 0.58 & 0.197 & 24.73 \\ 
        ~ & ~ & 0.3 & SAPS-Kn & 1 & 1 & 1 & 1 & - & 6 & 1 & 0.218 & 3.10 \\ 
        ~ & ~ & ~ & ACPS-Kn & 1 & 1 & 1 & 1 & - & 6 & 1 & 0.245 & 4.58 \\ 
        ~ & ~ & ~ & JDPS-Kn & 1 & 1 & 1 & 1 & - & 7 & 1 & 0.289 & 24.23 \\ 
        ~ & 50 & 0.1 & SAPS-Kn & 0.02 & 1 & 0.02 & 0.02 & - & 1 & 0.02 & 0.009 & 0.62 \\ 
        ~ & ~ & ~ & ACPS-Kn & 0.01 & 1 & 0.01 & 0.01 & - & 1 & 0.01 & 0.003 & 1.94 \\ 
        ~ & ~ & ~ & JDPS-Kn & 0.01 & 1 & 0.01 & 0.01 & - & 1 & 0.01 & 0.006 & 2.31 \\ 
        ~ & ~ & 0.2 & SAPS-Kn & 0.47 & 1 & 0.47 & 0.47 & - & 4 & 0.47 & 0.153 & 0.51 \\ 
        ~ & ~ & ~ & ACPS-Kn & 0.55 & 1 & 0.55 & 0.55 & - & 4 & 0.55 & 0.170 & 1.83 \\ 
        ~ & ~ & ~ & JDPS-Kn & 0.59 & 1 & 0.59 & 0.59 & - & 4 & 0.59 & 0.186 & 2.23 \\ 
        ~ & ~ & 0.3 & SAPS-Kn & 1 & 1 & 1 & 1 & - & 6 & 1 & 0.234 & 0.52 \\ 
        ~ & ~ & ~ & ACPS-Kn & 1 & 1 & 1 & 1 & - & 6 & 1 & 0.231 & 1.77 \\ 
        ~ & ~ & ~ & JDPS-Kn & 1 & 1 & 1 & 1 & - & 7 & 1 & 0.269 & 2.26 \\ \hline 
        $X_6$ & 20 & 0.1 & SAPS-Kn & 0.02 & 1 & 0.02 & 0.02 & 0 & 1 & 0 & 0.009 & 2.92 \\ 
        ~ & ~ & ~ & ACPS-Kn & 0.01 & 1 & 0.01 & 0.01 & 0 & 1 & 0 & 0.003 & 4.49 \\ 
        ~ & ~ & ~ & JDPS-Kn & 0.02 & 1 & 0.02 & 0.02 & 0 & 1 & 0 & 0.010 & 23.93 \\ 
        ~ & ~ & 0.2 & SAPS-Kn & 0.50 & 1 & 0.50 & 0.50 & 0.01 & 4 & 0.01 & 0.149 & 3.06 \\ 
        ~ & ~ & ~ & ACPS-Kn & 0.50 & 1 & 0.50 & 0.50 & 0.01 & 4 & 0.01 & 0.155 & 4.44 \\ 
        ~ & ~ & ~ & JDPS-Kn & 0.52 & 1 & 0.52 & 0.52 & 0 & 4 & 0 & 0.158 & 23.55 \\ 
        ~ & ~ & 0.3 & SAPS-Kn & 1 & 1 & 1 & 1 & 0.02 & 6 & 0.02 & 0.221 & 2.86 \\ 
        ~ & ~ & ~ & ACPS-Kn & 1 & 1 & 1 & 1 & 0.01 & 6 & 0.01 & 0.236 & 4.15 \\ 
        ~ & ~ & ~ & JDPS-Kn & 1 & 1 & 1 & 1 & 0 & 6 & 0 & 0.244 & 21.85 \\ 
        ~ & 50 & 0.1 & SAPS-Kn & 0.01 & 1 & 0.01 & 0.01 & 0 & 1 & 0 & 0.006 & 0.49 \\ 
        ~ & ~ & ~ & ACPS-Kn & 0.01 & 1 & 0.01 & 0.01 & 0 & 1 & 0 & 0.006 & 1.60 \\ 
        ~ & ~ & ~ & JDPS-Kn & 0.01 & 1 & 0.01 & 0.01 & 0 & 1 & 0 & 0.003 & 1.90 \\ 
        ~ & ~ & 0.2 & SAPS-Kn & 0.48 & 1 & 0.48 & 0.48 & 0.01 & 4 & 0.01 & 0.172 & 0.44 \\ 
        ~ & ~ & ~ & ACPS-Kn & 0.53 & 1 & 0.53 & 0.53 & 0.01 & 4 & 0.01 & 0.159 & 1.58 \\ 
        ~ & ~ & ~ & JDPS-Kn & 0.54 & 1 & 0.54 & 0.54 & 0.01 & 4 & 0.01 & 0.178 & 1.82 \\ 
        ~ & ~ & 0.3 & SAPS-Kn & 1 & 1 & 1 & 1 & 0.02 & 6 & 0.02 & 0.247 & 0.42 \\ 
        ~ & ~ & ~ & ACPS-Kn & 1 & 1 & 1 & 1 & 0.02 & 6 & 0.02 & 0.231 & 1.40 \\ 
        ~ & ~ & ~ & JDPS-Kn & 1 & 1 & 1 & 1 & 0.02 & 7 & 0.02 & 0.279 & 1.78 \\ \hline
    \end{tabular}
    }
\end{table}

Similar to Example 1, we evaluate the performance of three distributed feature screening methods based on knockoff features under Models (a) and (b). Tables \ref{tab3} and \ref{tab4} show the simulation results of the three methods in Model (a) under the setting of $(N,p,c)=(10000,5000, 0.725)$, where $\mathcal{P}_j$ denotes the proportion of the $j$th feature being selected over 200 replications. It can be seen that the three methods can successfully control the FDR under the pre-given level. Note that in Model (a), we have $0.1 < 1/|\mathcal{M}| <  0.2$. At the level of $\alpha=0.1$, the three methods sacrifice the sure screening property to control the FDR. At other levels of $\alpha$, the three methods guarantee the sure screening property and control the FDR simultaneously. Thus the simulation results are consistent with the conclusion of Theorem \ref{knock-th2}.

Table \ref{tab5} shows the results of the three methods in the setting of $(N,p,c)=(10000,10000,0.499 )$ in Model (b). When $Z=X_5$, $1/|\mathcal{M}|= 0.25$. Thus, with $\alpha=0.2$, the proportion that each important feature is selected is low, indicating that the sure screening property cannot be satisfied, which verifies Theorem \ref{knock-th2}. In addition, it can be seen that when $Z=X_6$, the existence of a nonlinear term makes the effect of the three screening methods worse, but all the methods still control FDR.  The sure screening property is not shown at the three levels of $\alpha$ in terms of the selection probability. one main reason is that the model contains a nonlinear feature, and the Pearson correlation is hard to describe the nonlinear relationship between features. We will leave the nonlinear problem in the future research.

\section{Real Data Analysis}
We apply our methods to the data YearPredictionMSD, which is available from the UCI Machine Learning Repository \citep{r21}.  As shown in the simulations, the ACPS method performs most satisfactorily for massive data set.  Also, each feature has a large number of repeated values in this data set, which limits the use of SAPS and JDPS methods.  Thus, we use the ACPS method for illustration in this real data set. The data set contains 515,345 songs released from 1922 to 2011, and is divided into training set (the first 463,715 samples) and test set (the last 51,630 samples) by the data provider. Each record in the data set contains the song release year and 90 audio features, where the audio features are taken from the Echo Nest API. One of the goals to analyzing this data is to predict the song release year through the audio features. Before the formal analysis, we perform some preprocessing on the data, including the standardization of data features and removing some abnormal sample data. We finally retained 445,200 training samples. In order to apply the proposed distributed feature screening method, we consider the interaction effect of all audio features, and finally obtain 4,095 features. We firstly select the feature with the largest Pearson correlation as the conditional feature, and then employ the ACPS-Kn method to filter the remaining 4,094 features. We next use the selected features to fit a linear model on the training set, and evaluate the prediction effect of the selected features by calculating the square root of the mean square prediction error of the model on the test set ($RMSE=(\frac{1}{n}\sum(Y_i-\hat{Y}_i)^2)^{1/2}$). 

We partition the data according to the order of the original data set and consider two cases of $K=200$ and $K=600$. As a comparison, we also give the results of feature screening based on the hard threshold, which corresponds to the fourth column in Table \ref{tab6}. The associated results of feature screening under different settings are reported in Table \ref{tab6}.

\begin{table}[htbp]
    \centering
    \caption{Performance of the model selected by ACPS method on the test set. The numbers in the brackets denote the model size.  } \label{tab6}
    \begin{tabular}{cccc}
    \hline
        $K$ & $\alpha=0.2$ & $\alpha=0.3$ & $\lfloor n/\log(n) \rfloor$ \\ \hline
        200 & 0.9315(5) & 0.9076(12) & 0.8753(288) \\ 
        600 & 0.9006(15) & 0.8983(21) & 0.8820(112) \\ \hline
    \end{tabular}
\end{table}

It can be seen from Table \ref{tab6} that although the RMSE on the test set is slightly higher than that based on the hard threshold, the model size based on the ACPS-Kn method is much smaller than that based on the hard threshold. This indicates that the ACPS-Kn method can reduce the number of false discoveries such that the selected features have a competitive prediction power. 

\section{Concluding Remarks}
This paper mainly studies massive ultrahigh-dimensional data. Firstly, three different distributed feature screening methods based on partial correlation are proposed, and the corresponding theoretical properties are established. Secondly, in the distributed framework considered in this paper, the threshold of the feature screening method is studied based on FDR control and knockoff features. A distributed feature screening algorithm that can control FDR is proposed, and the theoretical properties of the algorithm are proved. The numerical simulation results show that the proposed methods perform well, can control the FDR level under a given level, and the aggregated distributed method has better performance. Moreover, two extensions would be worth being further studied. First, this paper only considers the case where the conditional feature is one-dimensional, but in practical applications, the case where the conditional feature is multi-dimensional may be more common. Extending this to multiple conditional variables would be of interest. Second, this paper is based on Pearson correlation for feature screening, but Pearson correlation is difficult to effectively describe some nonlinear problems. So extending the Pearson partial correlation to other robust partial correlation is also of interest.

\bigskip 
\noindent\textbf{Acknowledgement} This work was partly supported by National Natural Science Foundation of China (Grant Number 11801202) and Fundamental Research Funds for the Central Universities (Grant Number 2021CDJQY-047).



\section*{Appendix: Proofs of Theorems}

\begin{proof}[Proof of Lemma \ref{knockoff-lemma}]
For convenience, we use $(\mathbf{X},\tilde{\mathbf{X}})_{(S)}$ to denote $(\mathbf{X},\tilde{\mathbf{X}})_{swap(S)}$, which simplifies to $(\mathbf{X},\tilde{\mathbf{X}})_{(j)}$ when $S=\{ j \}$, let $F_{y|x}(v|u)$ be the conditional distribution of $y$ given $x=u$. For $j \in \mathcal{M}^c$, we have
\begin{equation}
    F_{(Y,Z)|(\mathbf{X},\tilde{\mathbf{X}})_{(j)}}(v|(u,\tilde{u}))=F_{(Y,Z)|(\mathbf{X},\tilde{\mathbf{X}})}(v|(u,\tilde{u})_{(j)})= F_{(Y,Z)|\mathbf{X}}(v|u') \label{s-1},
\end{equation}
where $u'$ is the first $p$ element of $(u,\tilde{u})_{(j)}$, the second equality is due to the definition of knockoff features. Note that $j \in \mathcal{M}^c$, then we have
\begin{equation}
    F_{(Y,Z)|\mathbf{X}}(v|u')=F_{(Y,Z)|\mathbf{X}_{(-j)}}(v|u'_{(-j)})=F_{(Y,Z)|\mathbf{X}}(v|u)=F_{(Y,Z)|(\mathbf{X},\tilde{\mathbf{X}})}(v|(u,\tilde{u})) \label{s-2}.
\end{equation}
According to equations \ref{s-1} and \ref{s-2}, we have
$$ F_{(Y,Z)|(\mathbf{X},\tilde{\mathbf{X}})_{(j)}}(v|(u,\tilde{u}))=F_{(Y,Z)|(\mathbf{X},\tilde{\mathbf{X}})}(v|(u,\tilde{u})), $$
that is, $ (Y,Z)|(\mathbf{X},\tilde{\mathbf{X}})_{(j)}=_{d}(Y,Z)|(\mathbf{X},\tilde{\mathbf{X}}) $. By definition,  $(\mathbf{X},\tilde{\mathbf{X}})_{(j)}=_{d}(\mathbf{X},\tilde{\mathbf{X}})$, thus we have
$$ (Y,Z,(\mathbf{X},\tilde{\mathbf{X}})_{(j)})=_{d}(Y,Z,(\mathbf{X},\tilde{\mathbf{X}})), $$
namely, $ (Y,Z,X_j)=_{d}(Y,Z,\tilde{X_j}) $. Therefore, we have $\omega(Y,X_j,Z)=\omega(Y,\tilde{X_j},Z)$, implying that $\psi_j=0$. This completes the first part of the proof. Meanwhile, repeating the above steps, we have $ (Y,Z,(\mathbf{X},\tilde{\mathbf{X}})_{(S)})=_{d}(Y,Z,(\mathbf{X},\tilde{\mathbf{X}})) $, where $ S \subset \mathcal{M}^c $. Let $\hat{\psi}=(\hat{\psi}_1,\ldots,\hat{\psi}_p)^T$. We have
$$ \hat{\psi}=f( \mathcal{Y},\mathcal{Z},(\mathcal{X},\tilde{\mathcal{X}}) ), $$
where $f(\cdot): R^{N \times (2p+1)} \rightarrow R^{p}$. 
Let $\epsilon=(\epsilon_1,\ldots,\epsilon_p)$ be a sequence of independent random variables such that for $j \in \mathcal{M}$, $ P(\epsilon_j=1)=1 $, and for $j \in \mathcal{M}^c$, $ P(\epsilon_j=1)=P(\epsilon_j=-1)=1/2 $. For any $\epsilon$, let $S=\{ j:\epsilon_j=-1 \} \subset \mathcal{M}^c$, then we have
$$ (\hat{\psi}_1,\ldots,\hat{\psi}_p)^T=f( \mathcal{Y},\mathcal{Z},(\mathcal{X},\tilde{\mathcal{X}}) )=_{d}f( \mathcal{Y},\mathcal{Z},(\mathcal{X},\tilde{\mathcal{X}})_{(S)} )=(\epsilon_1 \hat{\psi}_1,\ldots,\epsilon_p \hat{\psi}_p)^T ,$$
where the last equality holds since $ f_j(\mathcal{Y},\mathcal{Z},(\mathcal{X}_j,\tilde{\mathcal{X}}_j))=-f_j(\mathcal{Y},\mathcal{Z},(\tilde{\mathcal{X}}_j,\mathcal{X}_j)) $ based on the definition of $\hat{\psi}_j$. Thus we have
$$ (\hat{\psi}_1,...,\hat{\psi}_p)^T=_{d}(\epsilon_1 \hat{\psi}_1,\ldots,\epsilon_p \hat{\psi}_p)^T. $$
This completes the second part of the proof\citep{r9}.
\end{proof}

In the following proof, we introduce some lemmas.
\begin{lemma}\label{b5}
Suppose that for $\theta_h$, $h=1,\ldots,s$, there exists a constant $a>0$ such that $|\theta_h|<a$. Let $\tilde{\theta}_h$ be an estimator of $\theta_h$ and suppose that for any $\epsilon>0$, there exists a constant $c>0$ such that, for any $h\in \left\{ 1,\ldots, s \right\}$,
$$ P( |\tilde{\theta}_h-\theta_h|\geq \epsilon )\leq c(1-\epsilon^2/c)^m.$$
Then, there exists a constant $c'>0$ such that
\begin{align*}
 P( ||\tilde{\theta}_h|-|\theta_h||\geq \epsilon ) & \leq c'(1-\epsilon^2/c')^m, \\
 P( |(\tilde{\theta}_{h_1} + \tilde{\theta}_{h_2})-(\theta_{h_1} + \theta_{h_2})|\geq \epsilon ) & \leq c'(1-\epsilon^2/c')^m ,\\
 P( |k\tilde{\theta}_h-k\theta_h|\geq \epsilon ) &\leq c'(1-\epsilon^2/c')^m ,\\
 P( |\tilde{\theta}_{h_1}\tilde{\theta}_{h_2}-\theta_{h_1} \theta_{h_2}|\geq \epsilon ) & \leq c'(1-\epsilon^2/c')^m.     
\end{align*}
If there exists a constant $b>0$ such that $|\theta_{h_2}|>b$, we have
$$ P( |\tilde{\theta}_{h_1}/\tilde{\theta}_{h_2}-\theta_{h_1} /\theta_{h_2}|\geq \epsilon )\leq c'(1-\epsilon^2/c')^m ,$$
and if $\theta_h>0$, we have
$$ P( |\sqrt{\tilde{\theta}_h}-\sqrt{\theta_h}|\geq \epsilon )\leq c'(1-\epsilon^2/c')^m .$$
\end{lemma}

\begin{proof}[Proof of Lemma \ref{b5}]
The proof follows from Lemma 6 of \citet{r2}.
\end{proof}

\begin{lemma}\label{b6}
Let $X_1,\ldots,X_n$ be a set of independent, centered random variables with $||X_i||_{\Psi_{\alpha}}\leq M$ for some $\alpha \in (0,1]$, where
$$ ||X||_{\Psi_{\alpha}}=\inf\left\{ t>0: E\exp(|X|^\alpha/t^\alpha)\leq 2 \right\}, $$
for any $a \in R^n$ and $t \geq 0$, we have
$$ P(|\sum^{n}_{i=1}a_iX_i|\geq t)\leq 2\exp(-\frac{1}{C_\alpha}\min( \frac{t^2}{M^2|a|^2} ,\frac{t^\alpha}{M^\alpha\max_{i}|a_i|^\alpha} )) , $$
where $C_\alpha$ is a positive constant that only depends on $\alpha$.
\end{lemma}

\begin{proof}[Proof of Lemma \ref{b6}]
The proof is adapted from Corollary 1.4. of \citet{r3}.
\end{proof}

\begin{lemma}\label{b7}
Suppose that there exists constants $(\kappa_0, D_0)$ such that, for any $0 \leq \kappa \leq \kappa_0$, $E\left\{ 
\exp(\kappa|X|)\right\}<D_0 $. Then, for any $0<\lambda<\kappa_0/(3e)$, we have 
$$ Ee^{\lambda X}\leq D_0\frac{\kappa_0}{\kappa_0-3e\lambda}. $$
\end{lemma}

\begin{proof}[Proof of Lemma \ref{b7}]
For any $t>0$ and $0<\kappa \leq \kappa_0$, we have
$$ P(|X|>t)=P(e^{\kappa|X|}>e^{\kappa t})\leq e^{-\kappa t}E\left\{ 
\exp(\kappa|X|)\right\}<D_0e^{-\kappa t}. $$
Thus, for $p\geq 1$, we have
\begin{align}
E|X|^{p}&=\int^{+\infty}_{0}P(|X|^{p}>u){\rm d}u=\int^{+\infty}_{0}P(|X|>t)pt^{p-1}{\rm d}t \nonumber \\
&\leq \int^{+\infty}_{0}D_0e^{-\kappa t}pt^{p-1}{\rm d}t=\int^{+\infty}_{0}D_0e^{-x}p(\frac{x}{\kappa})^{p-1}\frac{1}{\kappa}{\rm d}x \nonumber \\
&=\frac{pD_0}{\kappa^p}\Gamma(p)\leq \frac{pD_0}{\kappa^p}3p^p, \label{10-1}
\end{align}
where the last inequality holds since $\Gamma(x)\leq 3x^x$ when $x\geq 1/2$ \citep{r4}. Since \eqref{10-1} holds for any $0 < \kappa \leq \kappa_0$, we let $\kappa=\kappa_0$. When $\lambda>0$, we have
\begin{align}
Ee^{\lambda X}&=E(1+\sum_{n=1}^{\infty}\frac{(\lambda X)^n}{n!})=1+\sum_{n=1}^{\infty}\frac{\lambda^nEX^n}{n!} \nonumber \\
&\leq 1+\sum_{n=1}^{\infty}\frac{\lambda^n D_03^nn^n/\kappa_0^n}{(n/e)^n} \leq D_0\sum_{n=0}^{\infty}(\frac{3e\lambda}{\kappa_0})^n \nonumber \\
&=D_0\frac{1}{1-3e\lambda/\kappa_0}=D_0\frac{\kappa_0}{\kappa_0-3e\lambda}, \label{10-2}
\end{align}
where the first inequality is due to \eqref{10-1} and Stirling's formula ($n! \geq (n/e)^n$). This completes the proof.
\end{proof}

\begin{proof}[Proof of Proposition \ref{pro1}]
Note that 
\begin{align*}
\mathrm{var}(\bar{\theta}_{j,s})=\frac{1}{K}\mathrm{var}(\hat{\theta}^k_{j,s})=\frac{1}{N}\mathrm{var}(\hat{\theta}_{j,s}).
\end{align*}
By the condition \ref{a1}, $\mathrm{var}(\hat{\theta}_{j,s})$ is bounded uniformly. This completes the proof.
\end{proof}

\begin{proof}[Proof of Theorem \ref{b1}]
Since the following proof holds for any $j \in \{ 1,\ldots,p \}$ and $s \in \{1,\ldots,9\}$, we omit the subscript $(j,s)$ in some places for simplifying the notation. Let $\hat{\theta}$ denote a kernel of order $m$ of $\theta$. Then for a sample of size $n$ (denoted by $W$), a U-statistic of $\theta$ can be expressed as
$$ U_\theta=\binom{n}{m}^{-1}\sum_{i_1,\ldots, i_m}\hat{\theta}(W_{i_1},\ldots, W_{i_m}). $$
Let $r=\lfloor n/m \rfloor$ and $V_\theta=V_\theta(W_{i_1},\ldots, W_{i_n})=\frac{1}{r}\sum^r_{u=1}\hat{\theta}(W_{i_{(u-1)m+1}},\ldots, W_{i_{um}})$, then $U_\theta$ can be expressed as
$$ U_\theta=\frac{1}{n!}\sum_{i_1,\ldots,i_n}V_{\theta}(W_{i_1},\ldots,W_{i_n}),$$
where the summation is taken over different permutations of $\left\{1,\ldots, n\right\}$.

Now consider the distributed case. Let $U^k_\theta$ denote $U_\theta$ calculated based on $k$th subdata set. Define $\bar{U}_\theta=\frac{1}{K}\sum^K_{k=1}U^k_\theta$. For any $\epsilon>0$, $\nu>0$, we have
\begin{align}
 P(\bar{U}_\theta-\theta \geq \epsilon)&= P(e^{\nu(\bar{U}_\theta-\theta)} \geq e^{\nu\epsilon}) \nonumber \\
 &\leq e^{-\nu\epsilon}e^{-\nu\theta}Ee^{\nu \bar{U}_{\theta}}.  \label{10}
\end{align}  
According to the properties of convex function and Jensen's inequality, we have
\begin{align}
Ee^{\nu \bar{U}_{\theta}}&=Ee^{ \nu\{ \frac{1}{n!}\sum_{i_1,\ldots,i_n}\frac{1}{K}\sum^K_{k=1}V_{\theta}(W^k_{i_1},\ldots,W^k_{i_n}) \} } \nonumber \\
&\leq \frac{1}{n!}\sum_{i_1,\ldots,i_n}Ee^{ \frac{\nu}{K}\sum^K_{k=1}V_{\theta}(W^k_{i_1},\ldots,W^k_{i_n})  } \nonumber \\
&=Ee^{ \frac{\nu}{Kr}\sum^K_{k=1}\sum^r_{u=1}\hat{\theta}(W^k_{i_{(u-1)m+1}},\ldots, W^k_{i_{um}}) } \nonumber \\
&=(Ee^{ \frac{\nu}{Kr}\hat{\theta} })^{Kr}=(Ee^{ \kappa\hat{\theta} })^{Kr}, 
\end{align}
where $\kappa=\frac{\nu}{Kr}$. Thus,
\begin{equation}
P(\bar{U}_\theta-\theta \geq \epsilon)\leq (e^{-\kappa\epsilon}e^{-\kappa\theta}Ee^{ \kappa\hat{\theta} })^{Kr}. \label{12}
\end{equation}
By Taylor expansion, we have
\begin{align}
e^{-\kappa\theta}Ee^{ \kappa\hat{\theta} }&=Ee^{\kappa(\hat{\theta}-\theta)}=E[1+\kappa(\hat{\theta}-\theta)+\frac{1}{2}\kappa^2(\hat{\theta}-\theta)^2e^{\kappa_1(\hat{\theta}-\theta)}] \nonumber \\
&=1+\frac{1}{2}\kappa^2E[(\hat{\theta}-\theta)^2e^{\kappa_1(\hat{\theta}-\theta)}] \nonumber \\
&\leq 1+\frac{1}{2}\kappa^2\sqrt{E(\hat{\theta}-\theta)^4Ee^{2\kappa_1(\hat{\theta}-\theta)}} \nonumber \\
&\leq 1+2\kappa^2\sqrt{E\hat{\theta}^4Ee^{2\kappa_1\hat{\theta}}e^{-2\kappa_1\theta} }, \label{13}
\end{align}
where $\kappa_1 \in (0,\kappa)$, the last inequality is due to the $C_r$'s inequality and the Jensen's inequality. Lemma \ref{b7} shows that when $\nu$ is small enough, $2\kappa_1<\kappa_0/(3e)$. Furthermore, there exists $D_1>0$ such that $e^{-\kappa\theta}Ee^{ \kappa\hat{\theta} }\leq 1+D_1\kappa^2$, and $D_1$ is only dependent on $(\kappa_0,D_0)$.

Similarly, when $\kappa\epsilon<1$, there exists $D_2>0$ such that
$$
\begin{aligned}
e^{-\kappa \epsilon } e^{-\kappa \theta } Ee^{ \kappa\hat{\theta} } & \leq\left(1+D_1 \kappa^2\right)\left(1-\epsilon \kappa+D_2 \epsilon^2 \kappa^2\right) \\
& =1-\epsilon \kappa+D_2 \kappa^2 \epsilon^2+D_1 \kappa^2-D_1 \kappa^3 \epsilon+D_1 D_2 \kappa^4 \epsilon^2 \\
& \leq 1-\epsilon \kappa+D_2 \kappa^2 \epsilon^2+D_1 \kappa^2+D_1 D_2 \kappa^4 \epsilon^2 \\
& =1-\epsilon \kappa+E_1,
\end{aligned}
$$
where $E_1=D_2 \kappa^2 \epsilon^2+D_1 \kappa^2+D_1 D_2 \kappa^4 \epsilon^2$. 
Let $c_0=\kappa/\epsilon$, we have
$$
\frac{E_1}{\kappa \epsilon}=D_2 c_0 \epsilon^2+D_1 c_0+D_1 D_2 c_0^3 \epsilon^4 .
$$
Hence, for any $\epsilon>0$, there exists a small enough $\nu$ that makes $\kappa$ small enough. Thus $(\kappa_1,c_0)$ can be taken to be small enough to ensure $E_1/(\kappa\epsilon)<1/2,\kappa\epsilon<1$. Then,
\begin{equation}
e^{-\kappa \epsilon } e^{-\kappa \theta } Ee^{ \kappa\hat{\theta} }\leq 1-\kappa\epsilon/2 .\label{14}
\end{equation}
Besides, from \eqref{12} and \eqref{14}, we have
$$ P(\bar{U}_\theta-\theta \geq \epsilon)\leq (1-c_0\epsilon^2/2)^{K[n/m]}. $$
Similarly, $P(\bar{U}_\theta-\theta \leq -\epsilon)\leq (1-c_0\epsilon^2/2)^{K[n/m]}$.
It follows that for any $\epsilon>0$, there exists $c_0>0$ such that
$$ P(|\bar{U}_\theta-\theta| \geq \epsilon)\leq 2(1-c_0\epsilon^2/2)^{K[n/m]}. $$
Based on Lemma \ref{b5}, we have that for any $\epsilon>0$, there exists $c_2>0$ such that
$$ P(|\tilde{\omega}_j-\omega_j | \geq \epsilon)\leq c_2(1-\epsilon^2/c_2)^{N}. $$
Then, it follows that
\begin{align*}
P(\max_{1\leq j \leq p}|\tilde{\omega}_j-\omega_j | \geq \epsilon)&\leq \sum^p_{j=1}P(|\tilde{\omega}_j-\omega_j | \geq \epsilon)\\
&\leq pc_2(1-\epsilon^2/c_2)^{N}.
\end{align*} 
Thus, the second inequality in the theorem is proved.

Next, we consider the case of simple average. Similar to the above discussion, we can get $P(|\hat{\rho}_{k,j}-\rho_j | \geq \epsilon)\leq c_1(1-\epsilon^2/c_1)^{n}$. Then, we have
\begin{align*}
P(\max_{1\leq j \leq p}|\bar{\omega}_j-\omega_j | \geq \epsilon)&\leq \sum^p_{j=1}P(|\bar{\omega}_j-\omega_j | \geq \epsilon)=\sum^p_{j=1}P(||\frac{1}{K}\sum^K_{k=1}\hat{\rho}_{k,j}|-|\rho_j| | \geq \epsilon)\\
&\leq \sum^p_{j=1}\sum^K_{k=1}P(|\hat{\rho}_{k,j}-\rho_j | \geq \epsilon)\leq pKc_1(1-\epsilon^2/c_1)^{n}.
\end{align*}
Thus, the first inequality in the theorem is proved. 

Now consider the case of the JDPS method. Note that
\begin{align}
P(|\bar{\omega}^D_j-\omega_j | \geq \epsilon)&=P( || \frac{1}{K}\sum^{K}_{k=1}(\hat{\rho}_{k,j}-\hat{\Delta}_{k,j})|-|\rho_j| |\geq \epsilon )\nonumber \\
&\leq \sum^{K}_{k=1}P( |\hat{\rho}_{k,j}-\rho_j-\hat{\Delta}_{k,j}| \geq \epsilon)\nonumber \\
&\leq \sum^{K}_{k=1}P( |\hat{\rho}_{k,j}-\rho_j| \geq \epsilon/2)+\sum^{K}_{k=1}P( |\hat{\Delta}_{k,j}| \geq \epsilon/2),\label{15}
\end{align}
so we just need to consider $P(|\hat{\Delta}_{k,j}| \geq \epsilon)$. Let $D_i=(X_i,Y_i,Z_i,X_iY_i,X_iZ_i,Y_iZ_i,X_i^2,Y_i^2,Z_i^2)^T$ and $D=ED_i$, and let $\bar{D}_k=\frac{1}{n}\sum_{i \in \mathcal{S}_k}D_i$. The subscript $-i$ denotes the $i$th sample removed. We know that
\begin{align*}
\hat{\Delta}_{k,j}&=\frac{n-1}{n}\sum^n_{i=1}(\hat{\rho}_{k,j,-i}-\hat{\rho}_{k,j})=\frac{n-1}{n}\sum^n_{i=1}(g(\bar{D}_{-i,k})-g(\bar{D}_k))\\
&=\frac{n-1}{n}\sum^n_{i=1}\dot{g}(\xi_{i,k})^T(\bar{D}_{-i,k}-\bar{D}_k)=\frac{1}{n}\sum^n_{i=1}\dot{g}(\xi_{i,k})^T(\bar{D}_k-D_{i,k})\\
&=\sum^9_{l=1}\frac{1}{n}\sum^n_{i=1}\dot{g}(\xi_{i,k})_{(l)}(\bar{D}_{k(l)}-D_{i,k(l)}),
\end{align*}
where the subscript $(l)$ denotes the $l$th component of vector, and $\xi_{i,k}$ is on the line joining $\bar{D}_{-i,k}$ and $\bar{D}_k$. And we have
\begin{align}
P( |\hat{\Delta}_{k,j}| \geq \epsilon)&\leq P(|\hat{\Delta}_{k,j}| \geq \epsilon,\max_{i \in \mathcal{S}_k}|| \xi_{i,k}-D ||_2 \leq ||D||_2)\nonumber \\
&+P( \max_{i \in \mathcal{S}_k}|| \xi_{i,k}-D ||_2 > ||D||_2 )=:\Delta_1+\Delta_2 .
\end{align}
In what follows, we deal with $\Delta_1$ and $\Delta_2$ respectively. First consider $\Delta_1$. From the properties of continuous functions, we know that $\dot{g}$ is bounded on the region $|| \xi_{i,k}-D ||_2 \leq ||D||_2$. We denote its upper bound as $M_1$. Then,
\begin{align}
\Delta_1 &\leq P( |\sum^9_{l=1}\frac{1}{n}\sum^n_{i=1}\dot{g}(\xi_{i,k})_{(l)}(\bar{D}_{k(l)}-D_{i,k(l)})|\geq \epsilon , \max_{l \in \{ 1,\ldots,9 \}}\max_{i \in \mathcal{S}_k} |\dot{g}(\xi_{i,k})_{(l)}|\leq M_1) \nonumber \\
&\leq \sum^9_{l=1}P( |\frac{1}{n}\sum^n_{i=1}\dot{g}(\xi_{i,k})_{(l)}(\bar{D}_{k(l)}-D_{i,k(l)})|\geq \epsilon/9 , \max_{l \in \{ 1,\ldots,9 \}}\max_{i \in \mathcal{S}_k} |\dot{g}(\xi_{i,k})_{(l)}|\leq M_1 ) \nonumber \\
&\leq \sum^9_{l=1} \{ P( |\frac{1}{n}\sum^n_{i=1}\dot{g}(\xi_{i,k})_{(l)}(\bar{D}_{k(l)}-D_{(l)})|\geq \epsilon/18 , \max_{l \in \{ 1,\ldots,9 \}}\max_{i \in \mathcal{S}_k} |\dot{g}(\xi_{i,k})_{(l)}|\leq M_1 ) \nonumber \\
&+ P( |\frac{1}{n}\sum^n_{i=1}\dot{g}(\xi_{i,k})_{(l)}(\bar{D}_{i,k(l)}-D_{(l)})|\geq \epsilon/18 , \max_{l \in \{ 1,\ldots,9 \}}\max_{i \in \mathcal{S}_k} |\dot{g}(\xi_{i,k})_{(l)}|\leq M_1 ) \} \nonumber \\
&=: \sum^9_{l=1}( \Delta_{l1}+\Delta_{l2} ) ,\label{17}
\end{align}
and $\Delta_{l1}\leq P( M_1| \bar{D}_{k(l)}-D_{(l)} |\geq \epsilon/18 )=P( |\frac{1}{n}\sum^n_{i=1}(D_{i,k(l)}-D_{(l)})|\geq \epsilon/(18M_1) )$. 
At the same time, according to the condition \ref{a1}, there exists a constant $M_0$ which depends only on $(\kappa_0,D_0)$ such that $|| D_{i,k(l)}-D_{(l)}||_{\Psi_{1}}\leq M_0$. Thus, by Lemma \ref{b6}, we have
$$ \Delta_{l1}\leq 2\exp(-\frac{1}{C_1}\min(  \frac{n\epsilon^2}{(18M_0M_1)^2},\frac{n\epsilon}{18M_0M_1} ) ). $$
Similarly, we can get that $\Delta_{l2}$ satisfies the same probability inequality. Therefore, combining with \eqref{17}, we can get
$$ \Delta_1\leq 36\exp(-\frac{1}{C_1}\min(  \frac{n\epsilon^2}{(18M_0M_1)^2},\frac{n\epsilon}{18M_0M_1} ) ) .$$
Afterwards, we consider $\Delta_2$. Note that 
\begin{align}
\Delta_2 &\leq \sum^n_{i=1}P( || \xi_{i,k}-D ||_2 > ||D||_2 )=\sum^n_{i=1}P( || \xi_{i,k}-D ||^2_2 > ||D||^2_2 ) \nonumber \\
&=\sum^n_{i=1}P( \sum^9_{l=1}( \xi_{i,k(l)}-D_{(l)} )^2 > ||D||^2_2 )\leq \sum^n_{i=1}\sum^9_{l=1}P( | \xi_{i,k(l)}-D_{(l)} | > ||D||_2/3 )\nonumber \\
&\leq \sum^n_{i=1}\sum^9_{l=1} \{ P( | \bar{D}_{k(l)}-D_{(l)} | > ||D||_2/3 ) + 
P( | \bar{D}_{-i,k(l)}-D_{(l)} | > ||D||_2/3 ) \} \nonumber \\
&\leq \sum^n_{i=1}\sum^9_{l=1} \Big\{  2\exp(-\frac{1}{C_1}\min( \frac{n||D||^2_2}{9M^2_0},\frac{n||D||_2}{3M_0} ))  \nonumber \\
& \quad\quad + 2\exp(-\frac{1}{C_1}\min( \frac{(n-1)||D||^2_2}{9M^2_0},\frac{(n-1)||D||_2}{3M_0} )) \Big\} \nonumber \\
&\leq 36n\exp( -\frac{1}{C_1}\min( \frac{(n-1)||D||^2_2}{9M^2_0},\frac{(n-1)||D||_2}{3M_0} ) ) ,
\end{align}
thus $P( |\hat{\Delta}_{k,j}| \geq \epsilon/2)\leq 36\exp(-c_4n\min(c_5^2\epsilon^2,c_5\epsilon) )+36n\exp( -c_6(n-1) )$, where $c_4=1/C_1, c_5=1/(36M_0M_1)$ and $c_6=C_1^{-1}\min( (||D||_2/(3M_0))^2, ||D||_2/(3M_0) )$.
Using \eqref{15}, we have
\begin{align*}
& P(\max_{1\leq j \leq p}|\bar{\omega}^D_j-\omega_j | \geq \epsilon) \leq \sum^p_{j=1}P(|\bar{\omega}^D_j-\omega_j | \geq \epsilon)\nonumber \\
&\leq pKc_3(1-\epsilon^2/c_3)^n +36pK\exp(-c_4n\min(c_5^2\epsilon^2,c_5\epsilon) )+36pN\exp( -c_6(n-1) ).
\end{align*}
This completes the proof.
\end{proof}

\begin{proof}[Proof of Theorem \ref{b2}]
Let $\gamma=cN^{-\tau}$. If $\mathcal{M}\not\subseteq \bar{\mathcal{M}}$, that is, there exists $j \in \mathcal{M}$ such that $\omega_j \geq 2cN^{-\tau}$ 
and $\bar{\omega}_j<cN^{-\tau}$, then $|\bar{\omega}_j-\omega_j|>cN^{-\tau}$. We have
\begin{align*}
P( \mathcal{M}\subseteq \bar{\mathcal{M}} ) &\geq  P( \max_{j \in \mathcal{M}}| 
 \bar{\omega}_j-\omega_j|\leq cN^{-\tau} )\geq 1- P( \max_{j \in \mathcal{M}}| 
 \bar{\omega}_j-\omega_j|> cN^{-\tau} ) \\
 &\geq 1-|\mathcal{M}|P( |\bar{\omega}_j-\omega_j|>cN^{-\tau} )\geq 1-|\mathcal{M}|Kd_1(1-N^{-2\tau}/d_1)^n.
\end{align*}
Similarly, the rest of Theorem \ref{b2} can be proved. Furthermore, we note that
$$(1-N^{-2\tau}/d_1)^n=(1-1/(N^{2\tau}d_1))^{ N^{2\tau}d_1(n/(N^{2\tau}d_1)) }.$$
Thus, there exists $N_1$ such that when $N>N_1$, we have
$$ |\mathcal{M}|Kd_1(1-N^{-2\tau}/d_1)^n < pKd_1(\frac{2}{e})^{n/(N^{2\tau}d_1)}.$$
Let $p=O(e^{N^{v_1}}), K=O(N^{v_2})$, and $v_1,v_2>0,v_1<1-2\tau-v_2$. When $N$ is big enough, we have
$$ |\mathcal{M}|Kd_1(1-N^{-2\tau}/d_1)^n <d_6N^{v_2}e^{N^{v_1}-d_7N^{1-2\tau-v_2}}=o(1). $$
Similarly, under the same conditions, we can prove that both $|\mathcal{M}|d_2(1-N^{-2\tau}/d_2)^N$ and $|\mathcal{M}|Kd_3(1-N^{-2\tau}/d_3)^n+|\mathcal{M}|\delta(cN^{-\tau})$ are $o(1)$. Therefore, in the case of $ p=O(e^{N^{v_1}}) $, the proposed methods still enjoy the sure screening property.
\end{proof}

\begin{proof}[Proof of Theorem \ref{b3}]
Let $B=\{ j:\omega_j \geq \frac{1}{2}cN^{-\tau} \}$. It follows that $|B|\leq 2c^{-1}N^{\tau}\sum^p_{j=1}\omega_j$, otherwise we have
$$ \sum^p_{j=1}\omega_j\geq [1+2c^{-1}N^{\tau}\sum^p_{j=1}\omega_j]\frac{1}{2}cN^{-\tau}>(2c^{-1}N^{\tau}\sum^p_{j=1}\omega_j)\frac{1}{2}cN^{-\tau}=\sum^p_{j=1}\omega_j.$$
Let $C=\{ \max_{1\leq j \leq p}|\hat{\omega}_j-\omega_j|\leq \frac{1}{2}cN^{-\tau}  \},D=\{ j:\hat{\omega}_j\geq cN^{-\tau} \} $, then on the event $C$, we have $|B|\geq |D|=|\hat{\mathcal{M}}|$. Otherwise, there exists $j_0$ 
such that $\hat{\omega}_{j_0}\geq cN^{-\tau}$ and $\omega_{j_0} < \frac{1}{2}cN^{-\tau}$. Thus, $|\hat{\omega}_{j_0}-\omega_{j_0}|> \frac{1}{2}cN^{-\tau} $. It follows that
\begin{align*}
P( |\hat{\mathcal{M}}|\leq 2c^{-1}N^{\tau}\sum^{p}_{j=1}\omega_j )&\geq P(C)\\
&=1-P( \max_{1\leq j \leq p}|\hat{\omega}_j-\omega_j|> \frac{1}{2}cN^{-\tau} )\\
&=1-o(1).
\end{align*}
This completes the proof.
\end{proof}

\begin{proof}[Proof of Theorem \ref{b4}] It follows that
\begin{align*}
P( \min_{j \in \mathcal{M}}\hat{\omega}_j \leq \max_{j \in \mathcal{M}^c}\hat{\omega}_j )&=P( \min_{j \in \mathcal{M}}\omega_j - \max_{j \in \mathcal{M}^c}\omega_j - ( \min_{j \in \mathcal{M}}\hat{\omega}_j - \max_{j \in \mathcal{M}^c}\hat{\omega}_j )\geq \varDelta ) \\
&=P( (\max_{j \in \mathcal{M}^c}\hat{\omega}_j  - \max_{j \in \mathcal{M}^c}\omega_j )- ( \min_{j \in \mathcal{M}}\hat{\omega}_j - \min_{j \in \mathcal{M}}\omega_j )\geq \varDelta ) \\
&\leq P( \max_{j \in \mathcal{M}^c}| \hat{\omega}_j-\omega_j|+\max_{j \in \mathcal{M}}| \hat{\omega}_j-\omega_j| \geq \varDelta)\\
&\leq P( \max_{1\leq j \leq p}| \hat{\omega}_j-\omega_j| \geq \varDelta/2 ).
\end{align*}   
Thus $P( \min_{j \in \mathcal{M}}\hat{\omega}_j>\max_{j \in \mathcal{M}^c}\hat{\omega}_j )\geq 1- P( \max_{1\leq j \leq p}| \hat{\omega}_j-\omega_j| \geq \varDelta/2 )=1-o(1) $. This completes the proof.
\end{proof}

\begin{proof}[Proof of Theorem \ref{knock-th}]
In Lemma \ref{knockoff-lemma}, we have proved that $\hat{\psi}$ has the same property as the non-distributed estimation in \citet{r7}. The remaining proof can be completed by using the same arguments in the proof of \citet{r7}.
\end{proof}

\begin{proof}[Proof of Theorem \ref{knock-th2}]
Now we restrict ourselves to the sets $\{ \mathcal{D}_{(k)}^{(2)} ,k=1,\ldots,K\}$. Let $\mathcal{M}_{2}=\mathcal{M} \cap \hat{\mathcal{M}}_{1}$, $\mathcal{M}_{2}^{c}=\mathcal{M}^{c} \cap \hat{\mathcal{M}}_{1}$, $ \omega_j=\omega(Y,X_j,Z)$, and $\omega_j^{'}=\omega(Y,\tilde{X_j},Z) $. Taking the ACPS method as an example, it is known from the proof of Theorem \ref{b1} that
$$ P(|\hat{\omega}_j-\omega_j|\geq c(Kn_2)^{-\tau})\leq c_2(1-c_7(Kn_2)^{-2\tau})^{Kn_2}, $$
$$ P(|\hat{\omega}_j^{'}-\omega_j^{'}|\geq c(Kn_2)^{-\tau})\leq c_2(1-c_7(Kn_2)^{-2\tau})^{Kn_2}. $$
Thus
\begin{align*}
P(|\hat{\psi}_j-\psi_j| \geq 2c(Kn_2)^{-\tau} )&=P(|\hat{\omega}_j-\hat{\omega}_j^{'}-(\omega_j-\omega_j^{'})| \geq 2c(Kn_2)^{-\tau})\\
&\leq P(|\hat{\omega}_j-\omega_j| \geq c(Kn_2)^{-\tau})+P(|\hat{\omega}_j^{'}-\omega_j^{'}| \geq c(Kn_2)^{-\tau})\\
&= O((1-c_7(Kn_2)^{-2\tau})^{Kn_2}).
\end{align*}
Based on Lemma \ref{knockoff-lemma}(a), we have $\psi_j=0$ for $j \in \mathcal{M}^c$. Thus,
\begin{align*}
& P(\max_{j \in \mathcal{M}_2^{c}}|\hat{\psi}_j| \leq 2c(Kn_2)^{-\tau} )\\
&=P(\max_{j \in \mathcal{M}_2^{c}}|\hat{\omega}_j-\omega_j+\omega_j^{'}-\hat{\omega}_j^{'}| \leq 2c(Kn_2)^{-\tau})\\
&\geq P( \{\max_{j \in \mathcal{M}_2^{c}}|\hat{\omega}_j-\omega_j| \leq c(Kn_2)^{-\tau} \} \cap \{ \max_{j \in \mathcal{M}_2^{c}}|\omega_j^{'}-\hat{\omega}_j^{'}| \leq c(Kn_2)^{-\tau}\} )\\
&=1-P(\{\cup_{j \in \mathcal{M}_2^{c}}\{|\hat{\omega}_j-\omega_j| > c(Kn_2)^{-\tau} \}\} \cup \{ \cup_{j \in \mathcal{M}_2^{c}}\{|\omega_j^{'}-\hat{\omega}_j^{'}| > c(Kn_2)^{-\tau}\}\})\\
&\geq 1-\sum_{j \in \mathcal{M}_2^{c}}P(|\hat{\omega}_j-\omega_j| > c(Kn_2)^{-\tau} )-\sum_{j \in \mathcal{M}_2^{c}}P(|\omega_j^{'}-\hat{\omega}_j^{'}| > c(Kn_2)^{-\tau})\\
&\geq 1-O(n_{2}(1-c_7(Kn_2)^{-2\tau})^{Kn_2}).
\end{align*}
According to the conditions of Theorem \ref{knock-th2}, $\min_{j \in \mathcal{M}}\psi_j \geq 4c(Kn_2)^{-\tau} $, we have
\begin{align*}
P(\min_{j \in \mathcal{M}_2}\hat{\psi}_j \geq 2c(Kn_2)^{-\tau})&=1-P(\cup_{j \in \mathcal{M}_2}\{\hat{\psi}_j < 2c(Kn_2)^{-\tau}\})\\
&\geq 1-\sum_{j \in \mathcal{M}_2}P(|\hat{\psi}_j-\psi_j| \geq 2c(Kn_2)^{-\tau})\\
&\geq 1-O(n_{2}(1-c_7(Kn_2)^{-2\tau})^{Kn_2}).
\end{align*}
Thus, $ \min_{j \in \mathcal{M}_2}\hat{\psi}_j > \max_{j \in \mathcal{M}_2^{c}}|\hat{\psi}_j| $ holds with probability $ 1-O(n_{2}(1-c_7(Kn_2)^{-2\tau})^{Kn_2}) $. The other two distributed methods have similar conclusions, that is, important features are ranked before unimportant features with a probability approaching one. Note that we start with $\min_{j \in \hat{\mathcal{M}}_1}|\hat{\psi}_j|$ when finding the threshold based on equation \eqref{talpha} and then move to larger values until we reach the $|\hat{\psi}_j|$ satisfying equation \eqref{talpha}. Restrict on the event $\{ \min_{j \in \mathcal{M}_2}\hat{\psi}_j > \max_{j \in \mathcal{M}_2^{c}}|\hat{\psi}_j| \}$, and let $t_0=\min_{j \in \mathcal{M}_2}\hat{\psi}_j$. Then if $1/|\mathcal{M}| \leq \alpha$, we have
$$  \frac{1+\#\{ j :\hat{\psi}_j \leq -t_0 \}}{\#\{ j :\hat{\psi}_j \geq t_0 \}}=\frac{1+0}{|\mathcal{M}|} \leq \alpha \ .$$
Hence, $T_{\alpha}\leq t_0$, $\mathcal{M} \subseteq \hat{\mathcal{M}}(T_{\alpha})$. If $1/|\mathcal{M}| > \alpha$, in order to make the equation \eqref{talpha} valid, there must be $T_{\alpha} < t_0$. Also, the final selected model may contain the true model, or there may not be a threshold value satisfying the equation \eqref{talpha}.
\end{proof}

\end{document}